\nonstopmode
\documentclass[12pt]{article}

\input{structure/my-commands.tex}

\usepackage{caption} %% for changing spacing in captions
\usepackage{setspace}
\usepackage[safe]{tipa}
\usepackage[margin=1in]{geometry}
\usepackage{authblk} % for author's affiliations

\graphicspath{{figures/}}

\newcommand{\beginappendix}{%
        \setcounter{figure}{0}
        \renewcommand{\thefigure}{A\arabic{figure}}%
     }

\begin{document}

% \def\spacingset#1{\renewcommand{\baselinestretch}%
% {#1}\small\normalsize} \spacingset{1}

%%%%%%%%%%%%%%%%%%%%%%%%%%%%%%%%%%%%%%%%%%%%%%%%%%%%%%%%%%%%%%%%%%%%%%%%%%%%%%

  \title{\bf A Spatial Modeling Approach for Linguistic Object Data:\\ Analysing dialect sound variations
    across Great Britain}
    \author[1]{Shahin Tavakoli}
    \author[2]{Davide Pigoli}
    \author[3]{John A. D. Aston\thanks{ The authors gratefully acknowledge support
    from EPSRC grant EP/K021672/2}\footnote{Address for correspondence: Professor
    John Aston, Statistical laboratory, Department of Pure Mathematics and
  Mathematical Statistics, University of Cambridge, CB3 0WB Cambridge, United
Kingdom. Email: j.aston@statslab.cam.ac.uk} }
    \author[4]{John S. Coleman}
    \affil[1]{Department of Statistics, University of Warwick}
    \affil[2]{Department of Mathematics, King's College London}
    \affil[3]{Statistical Laboratory, University of Cambridge}
    \affil[4]{Phonetics Laboratory, University of Oxford}

   \maketitle

\bigskip
\begin{abstract}

Dialect variation is of considerable interest in linguistics and other social sciences. However, traditionally it has been studied using proxies (transcriptions) rather than acoustic recordings directly. We introduce novel statistical techniques to analyse geolocalised speech recordings and to explore the spatial variation of pronunciations continuously over the region of interest, as opposed to traditional isoglosses, which provide a discrete partition of the region. Data of this type require an explicit modeling of the variation in the mean and the covariance. Usual Euclidean metrics are not appropriate, and we therefore introduce the concept of $d$-covariance, which allows consistent estimation both in space and at individual locations. We then propose spatial smoothing for these objects which accounts for the possibly non convex geometry of the domain of interest.
We apply the proposed method to data from the spoken part of the British National Corpus, deposited at
the British Library, London, and we produce maps of the dialect variation over Great Britain. In addition, the methods allow for acoustic reconstruction across the domain of interest, allowing researchers to listen to the statistical analysis.

\end{abstract}

\noindent%

% \spacingset{1.63} % DON'T change the spacing!

\section{Introduction}

A better understanding of local dialect variation is of interest both from the point of view of
linguistics (how languages evolved in the past, how they became differentiated and how they will develop in the future) and from that of social sciences and demography, the way language is used being both a result of social affiliations and a tool to shape group identification. Dialect variations have long been
studied in sociolinguistics by considering textual differences between phonetic transcriptions of the words
\citep[see, e.g.,][and references
therein]{kretszchmar1996quantitative,nerbonne2003introducing,nerbonne2011gabmap}. This focus on written
forms reflects a general normative approach towards languages: for cultural and historical reasons, the
way we think about them is focused on the written expression of the words, even when thinking of their
pronunciations. However, this is more a social artifact than a reality in the population, as there is
great variation even within a single region with a claimed ``homogeneous'' dialect. Indeed, the analysis of
speech data highlights that the definition of language is an abstraction that simplifies the reality of speech
variability and neglects the continuous geographical spread of spoken varieties, although this does not exclude the presence of some clearly defined boundaries.

In this paper, we develop techniques that explicitly complement this text-based approach; we define methodology
not for written or transcription based analysis, but rather by treating the acoustic data directly, by
considering sounds as data objects \citep[see][for a definition of data objects]{wang2007object}. This allows
the examination of all forms of variation, including those within groups usually deemed to be homogeneous. To
achieve this requires the development of spatially varying statistical models for object data which take into
account both the underlying geography, but also the statistical properties of the data (in this case the fact
that part of the model requires estimation of quantities which lie on a manifold). This leads to the
definition of a new concept of covariance which is statistically consistent over space even under Fr\'echet
type estimation.

We are particularly interested in using information from speech recordings to model the smooth variation of
speech characteristics over a geographical region. Since recordings are obtained only in a discrete set of
locations, the first step will be the development of a non-parametric smoothing procedure to infer speech
characteristics (and plausible speech reconstruction) on the continuous map. Having available replicates from
different speakers at each location, we are able to model both the mean and the covariance structure of the
speech process at that single location, the latter being highlighted in recent studies
\citep[see][]{aston2010linguistic,hadjipantelis2012characterizing} as an important feature for language
characterization. The model we use to smooth the speech process over the whole geographical region of interest
is described in Section \ref{sec:model-and-estimation}, with the model based on the concept of using data
specific metrics in the analysis.

From a statistical point of view, we develop the concept of spatial object data analysis, and, in particular,
the use of $d$-covariances, that is, covariances that are estimated under a different metric to the usual
Euclidean ($L^2$) one. It has been seen in a variety of applications, particularly diffusion tensor imaging,
that even when the use of Euclidean distance is appropriate (and in the case of defining geodesics it may well
not be so), it is often sub-optimal in terms of interpretability \citep[see for example][]{dryden2009}. This
is particularly important for the case of spatial smoothing with replicates, as use of the implied Euclidean
metric (as is the case for the sample covariance) is not consistent with a spatially smoothed version under
another metric, while the Euclidean metric is not valid with general smoothing techniques for positive
definite covariances. Thus a new type of covariance will be developed which is statistically consistent. The
analysis of the covariance structure is made possible by the presence of replicates of the
same sound, uttered by different speakers, in each geographical location. This is an uncommon setting for
spatial data analysis which is usually focused on problems where replicates are not available but second-order
stationarity can be assumed. The latter is also the setting where most of the recent work on spatial
statistics for object data have been developed \citep[see, e.g.,][]
{delicado2010,gromenko2012,menafoglio2016}. In this work, the need to model the spatial variation of the
covariance structure of the speech process led us to choose a non-parametric regression approach to estimate
both the mean and the $d$-covariance of the speech process, in the line of the methods developed for
interpolation and smoothing of positive definite matrices \citep{dryden2009,yuan2012} and for surface
smoothing over complex domains \citep{wood2008,sangalli2013}.

We also develop a set of tools to communicate relevant information to linguists. First, we generate colour
maps that reflect speech variation in the spirit of isogloss maps \citep[see, e.g.,][]{francis1959some, upton2013} but
with continuous variation (as opposed to hard boundaries) and using information from speech recordings (as
opposed to achieving this via phonetic transcriptions). Moreover, our method allows the resynthesis of a
plausible pronunciation for any point in the considered geographical region.  We include as supplementary
material a few examples of these reconstructed pronunciations for the sound data set described in Section
\ref{sec:sounds}.

The paper proceeds as follows. In Section~\ref{sec:sounds}, the principles behind using acoustic
recordings as the intrinsic data objects, as well as the data set itself, are introduced.
Section~\ref{sec:model-and-estimation} develops both the concept of $d$-covariance and the model for
spatial data objects based on the $d$-covariance formulation.  Section~\ref{sec:results} applies the
modeling framework to the British National Corpus data. This data set is a large corpus of acoustic
recordings of British English across Great Britain, making it ideal for the comparison of dialects and
accents. Finally, Section~\ref{sec:discussion} is a discussion of the work and both its linguistic and
statistical relevance. 
Details on the data preprocessing, and technical results concerning the $d$-covariance and the model
are given in the Appendix.

%The contributions of this paper therefore include:
%\begin{itemize}
%\item Preprocessing and exploration of the spoken part of the British National Corpus.  The issues that arises in the use of speech data in dialectology (time/frequency representation, time misalignment, outlier detection) are addressed for these data but the approach can be extended to other speech dataset.
%\item Modeling of continuous speech variation other a geographical region, taking into account the irregular shape of the domain.
%\item Development of graphical outputs and interactive interface that help linguists in the analysis of the results from the model.
%\end{itemize}

\section{Sounds As Data Objects}\label{sec:sounds}

In linguistics, there has recently been a considerable interest in assessing information coming directly from speech
recordings \citep[][]{lehmann2004data,tfp2012phylogenetic,pigoli2014distances,hadjipantelis2015unifying,coleman2015} in addition to textual evidence and phonetic transcriptions. While we develop new methodologies that can
be applied to a variety of languages and geographical regions, we consider, in particular, the variation of the English
language in the United Kingdom. British English is well known to contain a large number of regional dialects, which can have considerable differences between them. Dialect variation is investigated by analysing the spoken part of the British National Corpus (BNC) deposited at the British Library. The
digital versions of these recordings are now made available by the Phonetics Laboratory of the University of
Oxford \citep{coleman2012audio}. These sound data (rather than their phonetic transcriptions) will be directly used to explore British dialects. In particular, for the statistical analysis of speech tokens, it is first necessary to represent sounds in a time-frequency domain and align them in time to account for individual variation in speaking rate.
We choose here a Mel Frequency Cepstral Coefficients representation for the speech tokens because of its good
performance for speech resynthesis (which will be the final output of our analysis), and because it provides a
principled lower dimensional representation of the speech tokens. We now give a more detailed description of the underlying data and their mathematical representation.

\subsection{Sound Waves, Spectrograms, and Mel-Frequency Cepstral Coefficients}
\label{sec:mfcc}
%\label{sec:spectrogram}
  A one-channel monophonic sound can be represented by a time series $\left( s(t) : t = 1,\ldots, T \right)$, where
  $s(t)$ represents the recording of the air pressure at time $t$ as captured by the microphone. As such, a
  sound is the variation of air pressure over time. For $t \leq 0$ or $t > T$, we let $s(t) = 0$. We can
  therefore assume that $s(t)$ is well defined for $t \in \bZ$. An example of sound wave is given in
  Figure~\ref{fig:sound-wave}.

  The spectrogram of a sound $(s(t))_{t = 1, \ldots, T}$ is a two-dimensional representation $\spec{s}(t, \omega)$ of
  the sound, where $\spec{s}(t', \cdot)$ represents the modulus of the discrete Fourier transform of $s(t)$
  in a neighborhood of $t'$. Mathematically, if $W(x), x \in \bR$ in a window function with support $[-1,1]$,
  then for any positive integer $M, w_M(t) = W(2t/M), t \in \bZ$ is a window of width $M$, and
  \begin{equation*}
    \spec{s}(t, \omega) = \left| \sum_{u = 1}^T s(t - u) w_M(u) \exp(- i \omega u) \right|, \quad t =
    1,\ldots, T, \omega \in [0, 2\pi].
    % \label{eq:spectrogram}
  \end{equation*}
  The function $u \mapsto s(t - u) w_M(u)$ is a windowed version of $s$ around $t$.
  For computational efficiency, the spectrogram is computed at the Fourier frequencies $\omega \in \left\{ 2 \pi
    k/N \right\}$, where $N \geq T$ is highly composite (usually a power of 2), using the fast Fourier transform
  \citep[FFT;][]{tukey:1965}. The window width $M$ is typically chosen to correspond to a segment of length
  ranging from $5$ to $20$ milliseconds ($M = 80$ or $320$ at 16Khz). From now on, we shall call the
  spectrogram of $s$ the $T \times N$ matrix with entries $ \spec{s}(t, \omega_k),  t = 1, \ldots, T,
  \omega_k = 2\pi k/N, k=0, \ldots, N-1$.

 %% The Cepstrogram $\cep{s}(t, q)$ is the inverse Fourier transform of the log spectrogram, i.e.
 %% \begin{equation*}
 %%   \cep{s}(t, q) = \frac{1}{N} \sum_{k = 0}^{N-1} \log\left[  \spec{s}(t, 2 \pi k/N) \right] \exp(i (2\pi q/N)k), \quad t=1,
 %%   \ldots, T; q = 0, \ldots, N-1.
 %%   % \label{eq:cepstrogram}
 %% \end{equation*}
 %% The variable $q$ is called a \emph{quefrency}, or \emph{cepstral coefficient}.
  A low dimensional time-frequency representation of the sound wave, often used in speech recognition and speech synthesis, are the \emph{Mel-frequency cepstral
    coefficients}, or \emph{MFCC}. The computation of the MFCC is done in two steps. First, the \emph{Mel
    spectrogram}, a filtered version
  of the spectrogram is computed,
  \begin{equation*}
    \melspec{s}(t, f) = \sum_{k=0 }^{N - 1} \spec{s}(t, 2 \pi k/N) b_{f, k}, \quad f = 0, \ldots, F,
    % \label{eq:filtered-spectrogram}
  \end{equation*}
  where $(b_{f, k})_{k = 0, \ldots, N-1}, f = 0, \ldots, F$ is the so-called Mel-scale filter bank (an example
  of a Mel-scale filter bank is given in \citealt{Gold2011}) with $F$ filters,
  which is believed to mimic the human ear auditory system. Then, the MFCC corresponds to the
  first $M \leq F$ coefficients of the inverse Fourier transform of the Mel spectrogram:
  \begin{equation*}
    \mfcc{s}(t, m) = \frac{1}{F} \sum_{f = 0}^F \log\left( \melspec{s}(t,f)  \right) \exp\left[ i (2\pi (m-1)/(F+1))f \right], \quad
    m=1, \ldots, M.
    % \label{eq:mfcc}
  \end{equation*}
  An additional reason to prefer MFCC over spectrograms is that each coefficient is
  associated to a frequency band, and therefore the MFCCs are more robust to small
  misalignments in frequency when comparing multiple speakers or sounds. Since the
  MFCCs are assumed to be smooth in $t$, we shall from now on assume that $t \in
  [0,1]$, where it is implicitly assumed that the integer $t$ are replaced by $t/T$
  and interpolated. An example of MFCC is given in Figure~\ref{fig:sound-wave}.

  Note that there exist many modifications and variations of this definition of MFCC
  in the literature, as authors seek improvements in the performance of
  implemented speech recognition or parametric speech synthesis systems. Since one of
  the goals of this paper is the resynthesis of sounds after inference, we shall use
  the definition and computational implementation of the MFCC proposed in
  \citet{Erro2011,Erro2014} as it yields high-quality, natural sounding resynthesised
  speech. However, the underlying principles are the same. For simplicity, we will
  refer in the following to this modified version as MFCC.

  \begin{figure}[p]
    \begin{center}
      \includegraphics{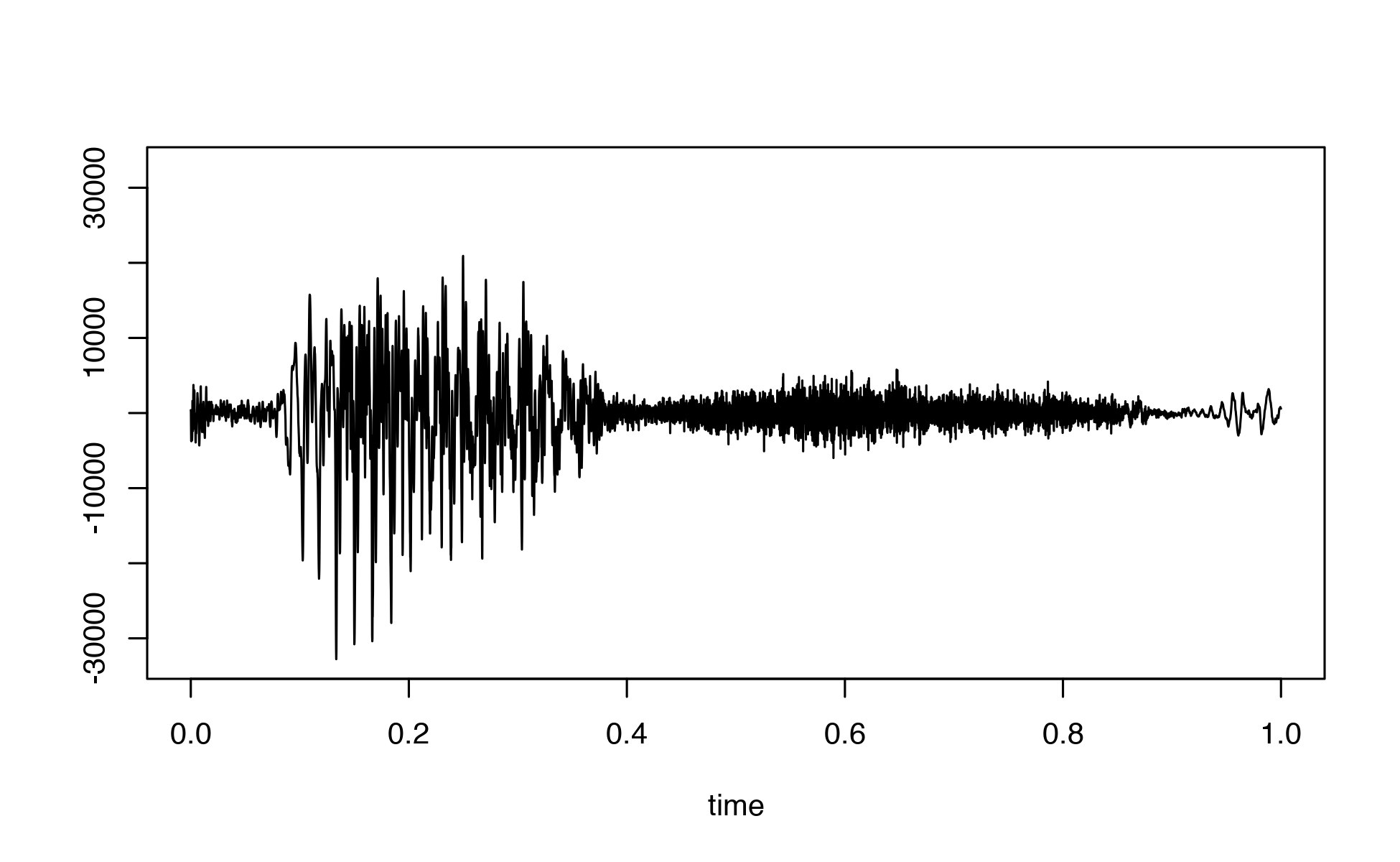}
      \includegraphics{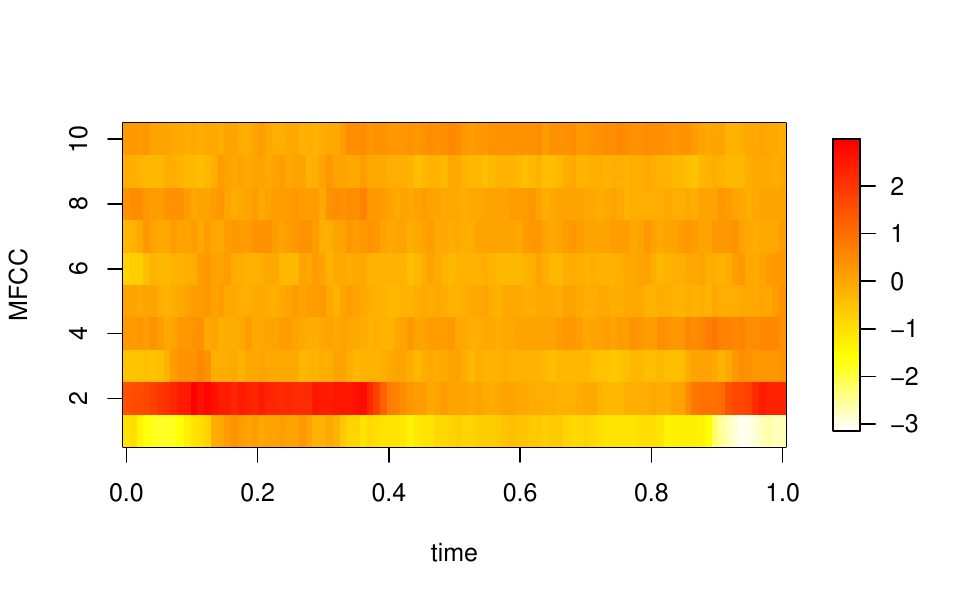}
    \end{center}
    \caption{Top figure: sound wave $s(t)$ of the word 'last'. Bottom figure: corresponding $\mfcc{s}(t,m)$. Both time
      scales have been normalized to $[0,1]$.}
    \label{fig:sound-wave}
  \end{figure}

\subsection{The British National Corpus}
\label{sec:bnc}

  The raw data consist of  the audio British National Corpus (BNC) recordings (131 GB of data, 16 Bit
  16 kHz one-channel .wav files, roughly 1100 hours of recording, publicly available at \url{http://www.phon.ox.ac.uk/AudioBNC}).
  These are mainly recordings of natural speech in typically noisy environments, with low recording amplitude
  (signal-to-noise ratio).
  Segmentation information about the words pronounced in the audio files were also provided (in TextGrid
  format), with the XML edition of the BNC (4.4 GB of files) containing transcriptions of the words spoken in the audio
  BNC recordings, along with contextual information (anonymized speaker identification, information about the speakers,
  location of the recording).

  For the purposes of the current paper, we restricted ourselves to the analysis of sounds of the vowel ``a''
  present in the
  following list of words:
  \begin{equation}
    %\texttt{bath, path, pass, class, glass, grass, past, last, brass, blast, ask, cast, fast}.
    \texttt{class, glass, grass, past, last, brass, blast, ask, cast, fast, pass}.
    \label{eq:list-of-words}
  \end{equation}
  The vowels in these words are pronounced in the same (geographically consistent) way and therefore we can
  consider them as the replicates of the same sound. We denote this as the ``class'' dataset. In Great Britain, this
  vowel is considered prototypical of the distinction between northern and southern accents: in the Midlands,
  North and South-West these words have a short, open front vowel [a] as in ``pat'', whereas in the South and
  South-East they have a long back vowel [\textipa{A}] (``aah''), similar to the vowel in ``part''.
The purpose of our work is the spatial analysis of sounds, and as such we needed to assign to each recorded
  sounds to the geographical location of the speaker's origin. We therefore removed sounds of speakers with
  missing or vague location information, and sounds corresponding to speakers who where trained to speak in
  a specified fashion (such as TV or radio presenters). For each speaker, we then used the corresponding
  recording location (the variable \texttt{placenamecleaned}) as a surrogate
  for the speaker's origin, provided this was unique. If there were multiple recording locations for a given
  speaker, the location corresponding to the \texttt{locale} variable ``home'' or ``at home'' was taken as the
  location of origin. If no such location existed, the sounds corresponding to the speaker were discarded. It
  should be noted that although we assume the speaker's accent to be representative of his recording location,
  there is unfortunately no data about the origin of the speakers to corroborate this assumption.

  After this process, we obtained  $4816$ sound tokens from $110$ distinct geographical locations in Great
  Britain, with $1993$ distinct speaker-vowel (from distinct word) combinations. About $46\%$ of the
  speaker-vowel combinations appear more than once, $61$ speaker-vowel combinations appear at least 10 times,
  and there is a speaker-vowel combination that is repeated $34$ times in the dataset. While the model we will use in Section \ref{sec:model-and-estimation} could be extended to a random effects type setup, to account for speaker repetition, we prefer to concentrate on the simpler model to aid understanding, particularly as most speaker vowel combinations only appear once.

  These vowel sounds were then transformed into MFCCs, with $M=10$. For each word $w$ in our list of words
  \eqref{eq:list-of-words}, we aligned the MFCCs of sounds corresponding to $w$ by registering their first
  coefficient (which corresponds to relative volume) using the Fisher-Rao metric \citep[R package \texttt{fdasrvf};
  see][]{srivastava2011registration,Tucker2013,Wu2014} and then extracted the segment associated to the vowel
  for each word, and linearly rescaled its time to the unit interval $[0,1]$.  Furthermore, we centered the first cepstral coefficients of all the data to remove
  differences in recording volume. For further information concerning the preprocessing and alignment, see the
  \supplement.

  We note that the duration of the vowel sounds is lost in the preprocessing step,
  and as such, we present a spatial map of the vowel durations in the \supplement
  (Figure~\ref{fig:sound-duration}). Indeed, the vowel duration is a one
  scalar summary of the sound of the vowel, which although useful does not capture considerable additional qualitative
  information contained in the vowel sound MFCCs, even after time alignment.

\section{Model and Estimation} \label{sec:model-and-estimation}

As mentioned earlier, previous works have identified the covariance structure between frequencies as an important feature of the
speech process that characterises languages
\citep{aston2010linguistic, hadjipantelis2012characterizing, pigoli2014distances, hadjipantelis2015unifying}. We
therefore have good reasons to expect the covariance between MFCCs---which are related to the energy in each
frequency band---to be associated with dialect characteristics, and we want to allow for it to vary
geographically. The investigation of the best metric for interpolation or extrapolation of covariance
matrices or operators has recently generated much work \citep{Arsigny2007,dryden2009,
  yuan2012,carmichael2013}. The use of a metric different from the Euclidean metric in the analysis leads to
the formulation of a more general concept of co-variability: the $d$-covariance. In the following, we explain
why we need to introduce this new concept and the role it plays in the definition of the model for speech
variation presented in Section~\ref{sec:model}.

\subsection{$d$-covariances}
\label{sec:d-cov}

Interpolation of covariance matrices under the usual Euclidean metric, although
yielding valid covariances, suffers from artifacts, such as swelling \citep[e.g.][]{arsigny2006log}. Extrapolation of
covariances under the Euclidean metric, on the other hand, is not even guaranteed to give valid covariances.
For this reason, several other metrics on the spaces of symmetric positive semi-definite matrices have been
studied, and have been shown to be useful for interpolation or extrapolation of covariances. 
For instance, the Euclidean average $\overline C = n^{-1} \sum_{i=1}^n C_i$ of covariance matrices $C_1,
\ldots, C_n$ can be reformulated as the solution to the variational problem
\[
  \min_{\Omega} \sum_{i = 1}^n d_E^2(\Omega, C_i),
\]
where $d_E$ denotes the Euclidean distance. In other words, $\overline C$ is the Fr\'echet mean of $C_1,
\ldots, C_n$ under $d_E$. Therefore the average of covariances $C_1, \ldots, C_n$ under another metric $d$ can be defined
as their Fr\'echet mean under $d$.

While covariance interpolation or extrapolation under various metrics is useful \citep[for example in the case of
spatial smoothing; see e.g.][]{yuan2012}, it is only valid when
treating the covariances $C_i$ as the observation units. However, the covariances $C_i$ are estimators of 
unknown true covariances, and have therefore an intrinsic estimation error. Since the true covariance of a
random vector $X \in \bR^p$ can be defined as the solution of the variational problem
\begin{equation*}
  \min_{\Omega} \ee d_E^2\left(\Omega, (X-\mu)(X-\mu)^\tp \right),
  % \label{eq:variational-problem-euclidean-covariance}
\end{equation*}
% based on observations $X_1, \ldots, X_n$ defined as the solution to
% \begin{equation}
  % \min_{\Omega} n^{-1} \sum_{i = 1}^n d_E^2\left( \Omega, (X_i - \overline X)(X_i - \overline X)^\tp \right),
  % \label{}
% \end{equation}
% where $\overline X = \sum_{i=1}^n X_i/n$.
 the sample covariance can be viewed as an analogue sample-based variational problem \emph{based on the
  Euclidean metric $d_E$}. When using a
metric different than $d_E$ for spatial smoothing of sample covariances, a consistency problem arises due
to the two different metrics used in the variational problem and the smoothing problem, and the resulting
estimator is biased. A one-dimensional
example illustrating this is given in Section~\ref{sec:example-d-cov} of the \supplement.
%% This happens because in-fill asymptotics (i.e.\ the number of locations going to infinity) and in-location
%% asymptotics (i.e.\ the number of observations per location going to infinity) are taken under different metrics.

For this reason, we introduce the concept of $d$-covariance that stems from recent developments on the
inference for covariance operators
\citep[see][]{arsigny2006log,dryden2009,kraus:2012,pigoli2014distances,petersen2016frechet} , where $d$ is a
metric on the space of $p \times p$ symmetric positive semi-definite matrices $\symmat_p$ that is used for the
spatial smoothing. The $d$-covariance of a random vector $X \in \bR^p$ is denoted $\cov_d(X)$, and defined by 
\begin{equation*}
  \cov_d(X) = \argmin_{\Omega \in \symmat_p} \ee{  d^2\left( ( X - \mu )(X - \mu)^\tp, \Omega \right) },
  % \label{eq:d-cov-definition}
\end{equation*}
where $\mu = \ee X$, and provided the right-hand side is well defined.

In this paper, we shall use the square-root metric $d_S$ on the space of symmetric positive
semi-definite matrices,
defined by $d_S(B,C) = \HSnorm{ \sqrt{B} - \sqrt{C} }$,  where $\sqrt B$, also written
$B^{1/2}$, is the unique square root of $B$ (meaning that it is the unique
matrix $D$ that satisfies $D D = B$; see Section~\ref{sec:square-roots} of the
\supplement),
and $\HSnorm{\cdot}$ is the Frobenius norm. As can be seen in the one-dimensional example given in
Section~\ref{sec:example-d-cov} of the \supplement, using the same metric $d_S$ for both the definition
of the co-variation and the spatial smoothing yields an estimator that is less biased than the one obtained by
spatial smoothing of the usual (Euclidean) covariance with $d_S$.
Let $\rpnorm{\cdot}$ denote the Euclidean norm on $\bR^p$, i.e.\ $\rpnorm{x} = \sqrt{x^\tp x}$. The
following Proposition gives an explicit formula for the
$d_S$-covariance.
\begin{prop}
  \label{prop:d-S-cov-well-defined-and-explicit-formula}
  Let $X \in \bR^p$ be random element with $\ee \rpnorm{X} < \infty$ and mean $\mu =
  \ee X$. Then $\cov_{d_{S}}(X) = \eee{ \sqrt{ (X - \mu)(X - \mu)^\tp } }^2$.
\end{prop}
Notice in particular that we do not need second moments for the $d_S$-covariance to exist, which is due to the
fact that $\HSnorm{\sqrt{XX^\tp}} = \rpnorm{X}$.
Since there is an explicit formula for
the square-root of symmetric positive semi-definite matrices of rank one, namely $\sqrt{xx^\tp} =
xx^\tp/\rpnorm{x}$, for $x \in \bR^p, x \neq 0$, we can rewrite the $d_S$-covariance of $X$ as
\begin{equation}
  \cov_{d_S}(X) = \eee{ \frac{(X - \mu)(X - \mu)^\tp}{\rpnorm{X- \mu}}}^2,
  \label{eq:ds-cov-as-regularized-cov}
\end{equation}
where the expression inside the expectation is understood to be equal to zero if $X =
\mu$. The denominator in \eqref{eq:ds-cov-as-regularized-cov} reveals that the square-root of the
$d_S$-covariance can be viewed as a regularized version of the usual covariance.
Furthermore, it also reveals that unlike the Euclidean covariance, the $d_S$-covariance does not behave in the
usual way under linear transformations: $\cov_{d_S}(AX)  \neq A \cov_{d_S}(X) A^\tp$ for general linear
transformations $A$. However, if we introduce the new families of square-root semi-metrics $d_{S, A}(C,D) =
\HSnorm{ \sqrt{ACA^\tp} - \sqrt{ADA^\tp}}$, where $C,D$ are $p \times p$ symmetric positive semi-definite
matrices and $A$ is a $n \times p$ matrix, we have the following result, proved in
Section~\ref{sec:technical-results} of the \supplement.
\begin{prop}
  \label{prop:d-covariance-and-linear-transformations}
  Let $A$ be a $n \times p$  matrix, and $X \in \bR^p$ be a random element with $\ee \rpnorm{X} < \infty$.
   Then $\cov_{d_S}(AX) = A \cov_{d_{S, A}}(X) A^\tp$.
   In the special case where $A^\tp A = I$, the identity matrix, we have
   $\cov_{d_S}(AX) = A \cov_{d_S}(X) A^\tp$.
\end{prop}
This means that the $d_S$-covariance of a linear transformation of $X$ is given by a
transformation of the $d$-covariance of $X$ under a metric related to the linear
transformation.
In particular, the entries of a $d_S$-covariance do not correspond to the $d_S$-covariance of corresponding entries of the random vector. This is analogous to partial correlation.
Furthermore,
Proposition~\ref{prop:d-covariance-and-linear-transformations} tells us that the
$d_S$-covariance is rotation equivariant.
Note also that the $d_S$-covariance is a measure of spread, and that other measure of
spreads have been proposed for multivariate data or functional data
\citep{Locantore1999a,Gervini2008,kraus:2012}, motivated from a robustness
perspective.

\subsection{A model for spatially varying speech object data} \label{sec:model}

We are now ready to define the model for speech variation for the analysis of
dialect data. We wish to have a model which can spatially vary both in terms of a
mean function but also in terms of covariance, as we will have replicates at
individual spatial locations. We therefore assume the following model:
\begin{equation}
  \label{eq:model}
  Y_{lj}(t) = m(X_l, t) + \varepsilon_{lj}(t), \quad l = 1,\ldots, L; j
  = 1,\ldots, n_l,
\end{equation}
where $Y_{lj}(t) \in \bR^p$ is the vector of the first $p$ Mel-frequency
cepstral coefficient (MFCC) at time $t \in [0,1]$ of the recording $lj$, $X_l$ corresponds to the spatial
location of the observations $Y_{lj}, j=1,\ldots, n_l$, recorded in latitude/longitude coordinates, i.e.\ $X_l
\in \england$, where $\england \subset \latlonspace$ is the spatial domain, and will denote Great Britain in the application of
Section~\ref{sec:results}.
% The $z_{lj} \in \mathcal Z$ corresponds to a covariate for observation
% $Y_{lj}$, and we assume that the set of covariates $\mathcal Z$ is finite.
% For each possible value of the
% covariate $z \in \mathcal Z$,
The \emph{spatial MFCC} is the function $x \mapsto m(x, \cdot) \in \mfccspace{p}$,
mapping a spatial location $x \in \england$  to its corresponding mean MFCC.

The term $\varepsilon_{lj} \in \mfccspace{p}$ is an error term. We assume that for each $l=1,\ldots, L$, $\varepsilon_{lj} \simiid
\varepsilon(X_l), j = 1,\ldots, n_l$, and that the $\varepsilon_{lj}s$ are all
independent. Indeed, this is a valid assumption since we have replicates for each location $X_l$, and a
scatterplot of the pairwise distances between the errors against their geographical distances does not reveal any spatial
dependence (see Figure~\ref{fig:scatterplot} of the \supplement).
  The process $\varepsilon(\cdot) : \england \rightarrow \mfccspace{p}$ is assumed to have mean
zero, $\ee \varepsilon = 0$, and we denote its $d_S$-covariance by $\Omega(x,t) =
\cov_{d_S}(\varepsilon(x, t))$, where we write $\varepsilon(x,t)$ for $\varepsilon(x)(t)$.
This implies in particular that $\cov_{d_S}(Y_{lj}(t)) = \Omega(X_l, t)$.
While traditionally $\Omega(X_l, t)$ would be defined as the covariance matrix of $Y_{lj}$, by assuming
that $\ee \varepsilon_{lj}(t) = 0$ for all $t$ and $\eee{ \varepsilon(t)
\varepsilon(t)^\tp }$ is the identity, we define here $\Omega(X_l, t)$ to be the
$d_S$-covariance of $Y_{lj}$,  where $d_S$ is the square-root metric, because we shall be smoothing spatially
using the metric $d_S$.
Recalling that $\symmat_p \subset \bR^{p \times p}$ is the space of symmetric positive semi-definite $p \times
p$ real matrices, the function $x \mapsto \Omega(x,\cdot) \in L^2([0,1], \symmat_p)$, maps a spatial location
$x \in \england$ to a time-varying symmetric positive semi-definite matrix at that location.

Given the observations $\{Y_{lj}(t),X_l\} $, we want to estimate a smooth field
$\widehat{m}(x, t) $ for the mean of the speech process, and a smooth field
$\widehat{\Omega}(x,t)$ for the (time-dependent) $d_S$-covariance between MFCCs coefficients.

\subsection{Estimation of the mean MFCCs field} 
\label{sec:mean-mfcc-field-estimation}

In this section, we will be dealing with the estimation of the mean MFCC field $m$, and
therefore the natural metric in this case to consider is the Euclidean ($L^2$)
distance. However, when we consider the geographical distance, the natural metric is the geodesic distance,
which we will approximate by graph distance $\graphdist(\cdot, \cdot)$ on a constructed triangular mesh over the region of interest.

We propose to fit the mean MFCC field using a local constant estimator which minimizes
a weighted mean square fit criterion. Let $K : \bR \rightarrow [0, \infty)$ denote
a continuous and  bounded density function,  and let
$K_h(s) = K(s/h)/h^2$.
At the location $x$, the estimate
of the mean MFCC is $\hat m(x) \in L^2\left( [0,1], \bR^p \right)$ which minimizes
%\begin{equation}
%  \sum_{l = 1}^L \sum_{j = 1}^{n_l} K_h\left( \graphdist(x, X_l) \right)
%  \frac{ \hnorm{Y_{lj} - \hat m(x) }^2 }{ \int_0^1 \trace{ \hat \Omega_l(t)} dt },
%  \label{eq:fit-criterion-mean-field}
%\end{equation}
\begin{equation}
  \sum_{l = 1}^L \sum_{j = 1}^{n_l} K_h\left( \graphdist(x, X_l) \right)
  \frac{ \hnorm{Y_{lj} - \hat m(x) }^2 }{ \hat \sigma^2(X_l)},
  \label{eq:fit-criterion-mean-field}
\end{equation}
where $\hnorm{\cdot}$ is the usual norm in $L^2\left( [0,1], \bR^p \right)$,
i.e.\ $\hnorm{f}^2 = \int_0^1 \rpnorm{f(t)}^2 dt$ for $f \in L^2\left( [0,1], \bR^p \right)$,
and
$\graphdist(x, X_l)$ is the distance on the map between $x$ and $X_l$. The denominator is a
normalizing factor that compensates for possible heteroscedasticity in the MFCC
field using the total variability of the residuals,
$\hat \sigma^2(X_l)= n_l^{-1} \sum_{j = 1}^{n_l} \hnorm{ Y_{lj} - \overline Y_l }^2$,
 where $\overline Y_l = n_l^{-1} \sum_{j = 1}^{n_l} Y_{lj}$.
The minimizer of the fit criterion
\eqref{eq:fit-criterion-mean-field} is a
Nadaraya--Watson type estimator, given by convex combination of the average MFCCs at
each location, i.e.\
\begin{equation}
  \hat m(x) = \sum_{l = 1}^L w_l(x) \overline Y_l,
  \label{eq:estimated-mean-field}
\end{equation}
where
\begin{equation*}
w_l(x) = \tilde w_l(x) / \sum_{l' = 1}^L \tilde w_{l'}(x) \quad \& \quad
\tilde w_l(x) = n_l \cdot K_h\left( \graphdist(x, X_l) \right) /\hat \sigma^2(X_l)
  % \label{eq:nadaraya-watson-weights-mean-field}
\end{equation*}

%%  We fit the spatial MFCC non-parametrically using a Nadaraya-Watson estimator, weighted by the inverse of the estimated total variance at the corresponding locations:
%%
%%  \[
%%   \label{eq:NWmean}
%%    \widehat{m}_{NW}(x)(t)=\sum_{l=1}^L \sum_{j}^{n_l}
%%\frac{K(\graphdist(X_l,x)/\hbar(x))/\int_0^1\mathrm{trace}(\widehat{\Omega}(x)(t)\mathrm{d}t)}{\sum_{l=1}^L
%%\sum_{j}^{n_l} K(\graphdist(X_,x)/\hbar(x))/\int_0^1\mathrm{trace}(\widehat{\Omega}(x)(t)\mathrm{d}t)} Y_{lj},
%%    \]
%%  where $K(.,.)$ denotes a kernel function, $\graphdist(X_l,x)$ is the distance on the map between the point of interest and the location $X_l$ and $\hbar(x)$ is a varying bandwidth that can change with the location $x$.

Possible strategies for the choice of the bandwidth $h$ are discussed in Section
\ref{sec:choice-smoothing-parameters}. 
It may be argued that using a higher order local polynomial
estimator in place of \eqref{eq:fit-criterion-mean-field} can reduce the bias of the estimator, and there
exists methods to perform local linear smoothing when only pairwise distances between the covariates are
available \citep{baillo2009local,boj2010distance,boj2016global}. We leave this extension as a future avenue of
research.
%% it is not
%% obvious how to define a local polynomial variation that respects the shape of the map and this shortcoming
%% tends to create problems close to the border if the region of interest is non-convex (as is the case with
%% Great Britain).

%However, in the
%supplementary material we describe an alternative to
%\eqref{eq:fit-criterion-mean-field} based on local linear polynomial.

\subsection{$d_S$-Covariance Field Estimation} 
In this section, we extend the kernel smoother to estimate the smooth $d_S$-covariance field $\Omega$. The
natural metric to be used for the smoothing in this case is the square root metric $d_S$ as this indeed avoids
inconsistencies between estimation of the $d_S$-covariance in the observed locations and estimation of the
spatially smooth field, as we show in Section \ref{sec:consistency}. Moreover, the square root metric is
well-defined for singular matrices, a property that will be needed for the application to the BNC data in Section
\ref{sec:results}. Indeed, locations with small number of observations are expected to have
$d_S$-covariance between MFCCs that are not full rank. We also propose to use a locally constant estimator of
the covariance field to allow for the non-convex domain, as discussed in the previous section.

At the point $x$, the estimated covariance $\hat \Omega(x, \cdot) \in L^2\left( [0,1],  \symmat_p \right)$ is the
minimizer of the following fit criterion:
\begin{equation}
  \sum_{l = 1}^L  K_h\left( \graphdist(x, X_l) \right)
  \int_0^1 d_S^2(\breve \Omega_l(t), \hat \Omega(x, t)) dt,
  \label{eq:fit-criterion-cov}
\end{equation}
where $h$ is a smoothing parameter,
and $\breve \Omega_l \in L^2\left( [0,1], \symmat_p \right)$ is the sample $d_S$-covariance at location $X_l$,
defined as
\begin{align*}
  \breve \Omega_l(t) &= \argmin_{\Omega \in \symmat_p} \: \frac{1}{n_l} \sum_{i = 1}^{n_l} d_S^2 \left(\Omega, (Y_{li}(t) - \overline
    Y_l(t))(Y_{li}(t) - \overline Y_l(t))^\tp \right), \quad \text{for each } t \in [0,1].
  \\ &= \left[ \frac{1}{n_l} \sum_{i=1}^{n_l} \sqrt{ (Y_{li}(t) - \overline Y_l(t)) (Y_{li}(t) - \overline Y_l(t))^\tp }
  \right]^2
\end{align*}
 % The coefficient $n_l$ in the fit criterion compensates for the unequal number of MFCCs used to
% compute the sample $d_S$-covariances: more weight is given to locations with more observations.
% The parameter $h$ is the smoothing bandwidth, which can depend on the location $x$.
It is not difficult to show (see Lemma~\ref{lma:soln-to-ds-cov-fit-criterion} of the \supplement) that the minimizer of \eqref{eq:fit-criterion-cov} is given by
\begin{equation}
  \hat \Omega(x,t) = \left[ \sum_{l = 1}^L w_l(x) \sqrt{\breve \Omega_l(t) }\right]^2,
  \label{eq:cov-est}
\end{equation}
where
\begin{equation}
w_l(x) = \tilde w_l(x) / \sum_{l' = 1}^L \tilde w_{l'}(x) \quad \& \quad
\tilde w_l(x) = K_h\left( \graphdist(x, X_l) \right) .
  \label{eq:nadaraya-watson-weights-cov-field}
\end{equation}
Equation \eqref{eq:cov-est} reveals that $\hat \Omega(x)$ is the square of a Nadaraya--Watson estimator in the
square-root space.
%\begin{align}
%  \label{eq:cov-est}
%  \widehat{\Omega}_{NW}(x)(t)&=(\widehat{\Omega^{1/2}})^T(\widehat{\Omega^{1/2}}), \\
%  \nonumber
%  \widehat{\Omega^{1/2}}&=\sum_{l=1}^L \sum_{j}^{n_l}
%  \frac{K(\graphdist(X_l,x)/\hbar(x))}{\sum_{l=1}^L \sum_{j}^{n_l} K(\graphdist(X_,x)/\hbar(x))/} (S(X_l)(t))^{1/2} \\
%\end{align}
% where $S(X_l)(t)$ denotes the estimated covariance matrix at the observed location $X_l$ at time $t$,
% and $(S(X_l)(t))^{1/2}$ its square root and $\graphdist(.,.)$ denotes the geographical distance. $\hbar(x)$ is a varying
% bandwidth which may depend on $x$ and it can be chosen analogously to the bandwidth for the mean estimator as
% described in Section \ref{sec:choice-smoothing-parameters}.

% $S(X_l)(t)$ can be obtained either as the sample covariance matrix of the observations at $X_l$ or updated with the estimated smooth mean as
% \[
  % S(X_l)(t)=\sum_{j=1}^{n_l}(Y_{lj}- \widehat{m}(X_l)(t))(Y_{lj}- \widehat{m})(X_l)(t))^T.
% \]

\subsection{Consistency of Smoothing with the Square Root Distance}
\label{sec:consistency}

In this section, we study the properties of the estimator for the $d_S$-covariance smooth field $\Omega(x,t)$.
This is a non-standard smoothing problem, which poses a few theoretical challenges due to the
non-Euclidean metric involved, and to the fact that we want to control the estimation error uniformly in the
time index. Moreover, this gives us the opportunity to show how it is possible to account
for
the use of the geographical distance $\graphdist$ in the kernel smoothing. The estimator for the mean field
$m(x,t)$ uses the Euclidean ($L^2$) metric and its properties can therefore be studied using similar arguments, in
particular using the results in Section~\ref{sec:technical-results} of the
\supplement.

This first result, proved in Section~\ref{sec:technical-results} of the
\supplement, shows that under mild assumptions, the sample $d_S$-covariance is a $\sqrt{n}$-consistent
estimator of the $d_S$-covariance.
\begin{prop}
  \label{prop:rate-of-convergence-sample-ds-cov}
  Let $Y_1, \ldots, Y_n \simiid Y \in \bR^p$ be random vectors with $\mu = \ee Y$ and
  $\ee \rpnorm{Y}^2 < \infty$. Let $\overline Y = (Y_1 + \cdots + Y_n)/n$, and
  \begin{align*}
    \breve \Omega &= \left( \frac{1}{n} \sum_{i = 1}^n \sqrt{ (Y_i - \overline Y)(Y_i -
\overline Y)^\tp } \right)^2,
  \end{align*}
  this being the explicit expression for the sample $d_S$-covariance.
  Then $d_S( \breve \Omega, \cov_{d_S}(Y) ) \leq \kappa_p \sqrt{ \ee \rpnorm{Y - \mu}^2/n }$, where $\kappa_p$ is a constant depending only on the dimension.
\end{prop}
In particular, $ \breve \Omega = \cov_{d_S}(Y) + \bigop(n^{-1/2})$.
We now introduce some conditions used in proving the consistency of the smooth $d_S$-covariance field.
\begin{cond}
  \label{cond:kernel-and-graph-distance}
  \mbox{}

  \begin{enumerate}[(1)]
    \item \label{cond:kernel} The kernel $K: \mathbb R \rightarrow [0, \infty)$ is a continuous probability density,
       with $\int_{0}^\infty s^3 K(s) ds < \infty$.
      Assume also that $K$ is decreasing, i.e.\ $0 \leq s \leq t \implies K(s) \geq K(t)$.
    \item \label{cond:graphdist-equiv-to-euclidean} There exists constants $0 < c_1 < c_2$ such that $c_1 \rpnorm{x-y} \leq
      \graphdist(x,y) \leq c_2 \rpnorm{x-y}$.
  \end{enumerate}
\end{cond}
Condition~\ref{cond:kernel-and-graph-distance} (\ref{cond:kernel}) is a standard condition on the kernel function, which is
in particular satisfied by the Gaussian kernel we use in Section~\ref{sec:results}.
Condition~\ref{cond:kernel-and-graph-distance} (\ref{cond:graphdist-equiv-to-euclidean}) states that the graph distance is (metric) equivalent to the
Euclidean distance on $\england$.
The following condition on the sampling density is standard.
\begin{cond}
  \label{cond:density}
  The density of the observation locations $X_1, \ldots, X_L \in \england$,  $f : \england \rightarrow
    \bR$ is continuous, and $\sup_{x \in \england} f(x) < \infty.$
\end{cond}
Recall that $\Omega(x,t) = \cov_{d_S}(\vep(x, t))$. We are going to assume the following regularity
conditions on the error process $\vep$.
\begin{cond}
  \label{cond:consistency-dcov-smoothing}
  \mbox{}

  \begin{enumerate}[(1)]
  \item $n_l \geq c_0 n$ for all $l = 1,\ldots, L$ for some $c_0 > 0$.
  \item $\sqrt{\Omega(\cdot,t)} : \england \rightarrow
  \symmat_p $ is $C^1$ (with respect to the Hilbert--Schmidt norm), and \\$\sup_{x \in \england, t \in
    [0,1]}\rpnorm{ \frac{\partial \left[ \sqrt{\Omega(x,t)} \right]_{rs}}{ \partial
      x}(x) } < \infty$, where $[A]_{rs}$ is the $rs$-th entry of the matrix $A$.
    \item $\sup_{x \in \england, t \in [0,1]} \ee{ \rpnorm{\vep(x, t) }^2 } <
      \infty$.
  \end{enumerate}
\end{cond}
Condition~\ref{cond:consistency-dcov-smoothing} (1) states that asymptotically, the number of observations per
locations is of the same order.
Condition~\ref{cond:consistency-dcov-smoothing} (2) is a (pointwise in time) smoothness condition on the
$d_S$-covariance field.
Condition~\ref{cond:consistency-dcov-smoothing} (3) assumes that the second moment of the error field is
uniformly bounded. The second moment is needed to establish the rate of convergence, whereas the uniform bound
is for the control of the smoothing error uniformly in time.

We can now state the result on the consistency of the smoothed $d_S$-covariance field, whose proof is in
Section~\ref{sec:technical-results} of the
\supplement.
\begin{thm}
  \label{thm:consistency}
  Assume model \eqref{eq:model} with conditions~\ref{cond:kernel-and-graph-distance}, \ref{cond:density} and
  \ref{cond:consistency-dcov-smoothing} holds, and $L \rainf, h \raz, Lh \rainf, n \rainf$.
  Then, for any $x \in \england$ in the interior of $\england$ such that $f(x) > 0$, we have
  \begin{equation*}
    \ee_X d_{S}(\hat \Omega(x,t), \Omega(x,t))  = \bigop(n^{-1/2}) + \bigop\left( \sqrt{h^2 +
        \frac{1}{nLh^2}} \right),
    % \label{eq:consistency-rate-of-conv}
  \end{equation*}
  where the stochastic term is uniform in $t$, and $\ee_X$ is the expectation conditional on $X_1, \ldots,
  X_L$.
\end{thm}
The first error term comes from the fact that we are using the sample mean in place of the true mean in the
computation of the sample $d_S$-covariance, while the second error term is a bias plus variance decomposition. Notice
that the $n$ in the variance term is unusual, and is related to the estimation error of $d_S$-covariances at
the observation locations. In particular, the variance is inversely proportional to the number of observations
per location, regardless of $L$ and $h$.

\section{Analysis of sound data from the BNC}
\label{sec:results}

We apply here the proposed method to the ``class'' dataset described in Section~\ref{sec:sounds}.
 %i.e. the MFCCs of the vowel sounds from instances of the words `fast'', ``past'', ``cast'', ``last'',``class'', ``glass'', ``grass'', ``pass'', ``brass'' and ``ask'' in the spoken part of the BNC.
As mentioned, the sound tokens come from $110$ distinct locations within Great Britain, which are
indicated on the geographical map in Figure~\ref{fig:aast-loc}. Also shown is the triangulation used for the
smoothing, where the internal nodes contain the locations of the observations.
It can be seen that the observed locations are irregularly spaced in the region, with high density of the
observations around London and other large cities and very sparse observations in Wales and central Southern
England, for example. In particular, only three locations are available in Scotland and therefore we will not
draw strong conclusions about the dialect variation in that country. Figures~\ref{fig:gb-counties} and
\ref{fig:gb-regions} in the \supplement show the counties and regions of Great Britain.

While the method allows for a smooth reconstruction of the sound from the mean MFCC (and a few of these
reconstructed sounds for the vowel described above can be found as Supplementary Material), we want also to
represent the sound variations (and those of the their $d_S$-covariance) on a map to be able to explore dialect
variations. We need therefore to reduce the dimensionality of the data object. Among the possible
alternatives, we choose to project the mean smooth field onto the principal components obtained from the
original data because this allows the comparison of the projections of the smooth field estimated using different
choices of bandwidth parameters. Concerning the $d_S$-covariance smoothed field, its variation will be explored
by considering the pairwise distances of the estimated $d_S$-covariance field, and comparison of
 these distances with those obtained under the assumption that there
is no spatial variation in the $d_S$-covariance field. We will use a Gaussian kernel for all
the results of this Section, and we will now discuss the choice of the smoothing parameters.

\begin{figure}[th]
  \centering
    \includegraphics[page=1]{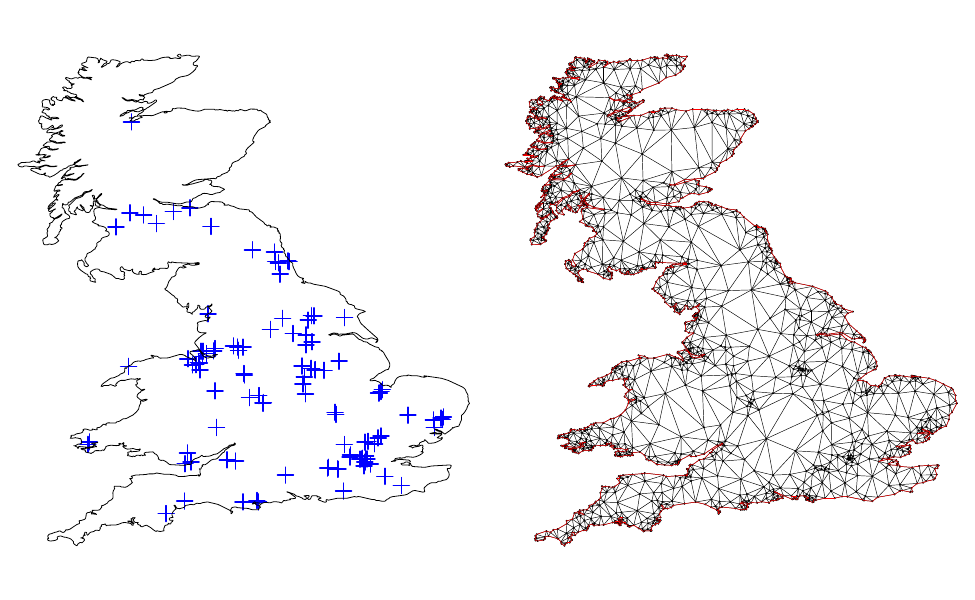}\\
   \caption{Geographical locations where data are available in the ``class'' dataset (left) and triangulation
     of Great Britain (right).}
  \label{fig:aast-loc}
\end{figure}

\subsection{Choice of Smoothing Parameters} \label{sec:choice-smoothing-parameters}

\subsubsection{Varying Bandwidths}
In both the estimator for the mean and the covariance, a bandwidth varying with the geographical
location can be used. This is particularly important when the locations of the observations are irregularly spaced in
the region of interest, as is the case for the ``class'' dataset, where the use of a constant bandwidth would
lead to over- or undersmoothed estimates.
% oversmooth where observed locations are
% dense or not having any observation in the support of the kernel function in the regions where observed
% locations are sparser.

A possible approach is to adapt the bandwidth to the density of the observations
 using the distance from the $k$-th nearest \emph{location} to modulate the global bandwidth, i.e.\
 \begin{equation}
   \label{eq:band-location}
   \hbar(x)=h \cdot \Delta_L(x, k),
   %\argmin{\Delta}{\sum_{l=1}^{L} \mathds{1}_{\{d(x,X_l)\leq \Delta\}}\geq k},
 \end{equation}
 where $\Delta_L(x, k)$ is the geographical distance between $x$ and the $k$-th nearest location to $x$.
 This \emph{$k$-th-nearest locations varying bandwidth} adjusts the bandwidth to the density of the observed
 locations, thus guaranteeing that information from comparable numbers of observed locations are used in the
 estimation at each point. However, if there is a large variability between number of observations at
 different locations, one may prefer to adjust to the number of observations. We can then define a
 \emph{$k$-th-nearest observations varying bandwidth} as
\begin{equation}
  \label{eq:band-observation}
  \hbar(x)=h \cdot \tilde \Delta(x, k),
  %\argmin{\Delta}{\sum_{l=1}^{L} n_l \mathds{1}_{\{d(x,X_l)\leq \Delta\}}\geq k'}.
\end{equation}
where $\tilde \Delta(x, k)$ is the least distance from $x$ within which there are at least $k$ observations,
i.e.\
\[
  \sum_{l \, :\,  d(x, X_l) < \tilde \Delta(x, k) }  n_l < k \qquad \text{and}
  \qquad \sum_{l \, :\,  d(x, X_l) \leq \tilde \Delta(x, k)} n_l \geq k.
\]
A third alternative would be to simply using a fixed bandwidth $ \hbar(x)=h$ for all $x$, but this leads to
the problem of oversmoothing in the regions with denser observations, as mentioned above.

The idea of adjusting the bandwidth on the basis of observation density is well known in non-parametric
regression \citep[see e.g.][]{fan1995data}, but the difficulty in estimating the bivariate density with a
relatively small numbers of observations led us to prefer the use of the distance from the $k$-th nearest
neighbour as proxy for the inverse of the density of the observations, this distance being expected to be
small in high density regions and large in low density regions.

The expressions of the bandwidth in \eqref{eq:band-location} and \eqref{eq:band-observation} contain two
parameters that need to be chosen: the number $k$  of nearest neighbours to be used to adapt
the bandwidth and the global smoothing parameter $h$. These can be chosen by cross-validation, as described in
the next Section.
% , or by visual
% inspection of the maps obtained through low-dimensional projection of the estimated field, as is
% described in Section~\ref{sec:projection-pc}.

\subsubsection{Cross-validation for varying bandwidth parameters}
% \label{sec:cross-validation}

The choice of the parameters $k$ and $h$ for the varying bandwidths \eqref{eq:band-location} and
\eqref{eq:band-observation} can be guided by estimating the prediction error as a function of such
parameters using a cross-validation procedure. We propose here to use a
leave-one-location-out cross validation for the choice of the parameters $k$ and $h$. For the mean field, the
cross-validation is defined by
\begin{equation*}
  \text{mean.cv}(k,h) = \sum_{l = 1}^L \frac{ \hnorm{ \overline Y_l - \hat m_{-l}(X_l) }^2 }{\hat{\sigma}^2(X_l)},
  % \label{eq:mean-cv}
\end{equation*}
where $\hat m_{-l}$ is the estimate of the MFCC field obtained without all the MFCCs observed at location $X_l$.
Analogously, we can define a cross-validation error for the $d_S$-covariance estimator as
\begin{equation*}
  \text{cov.cv}(k,h) = \sum_{l = 1}^L \int_{0}^1 d_S^2(\widehat \Omega_{-l}(X_l)(t),\breve{\Omega}_l(t)) \mathrm{d}t,
  % \label{eq:cov-cv}
\end{equation*}
where $\widehat \Omega_{-l}(X_l)$ is the prediction for the $d_S$-covariance at location $X_l$
obtained from (\ref{eq:cov-est}) without the observations at location $X_l$, and $\breve{\Omega}_l$ is the sample $d_S$-covariance at $X_l$.
%The constant $p$ is the number of MFCC coefficients and is here only for consistency with figures shown below.
It is however important also to explore the results visually, using the strategies described in Section~\ref{sec:projection-pc}, for different values of the
smoothing parameters to be sure that the chosen parameters are not leading to oversmoothing or overfitting.

For the ``class'' dataset, the cross-validation errors different values of $h$ and $k$ can be found  in
Figure~\ref{fig:aast-cv-loc} (for the nearest locations bandwidth) and in
Figure~\ref{fig:aast-cv-obs} of
the \supplement (for the nearest observations bandwidth).

For the mean field, the nearest locations bandwidth yields the minimal cross-validation errors, with  $h=0.5,\, k=14$.
The nearest observations bandwidth yields a slightly higher minimal cross-validation error ($h=1.5, k=300$).
As will be seen in Figures~\ref{fig:aast-mean-optim-nearest-loc} and
Figure~\ref{fig:aast-mean-optim-nearest-obs} of the \supplement,
the nearest location bandwidth yields maps that capture more of the variability of the mean MFCC field,
whereas the nearest observations bandwidth seems to be oversmoothing. Therefore, we shall use the varying
bandwidth with the nearest locations for the interpretation.

For the $d_S$-covariance field, the cross-validation curves decrease as the bandwidth parameter $h$
increases, and seem to reach a plateau. We take the smallest value of $h$ (and the corresponding $k$) that
reaches the plateau, which is $h=1, k=32$ nearest locations. Although this choice
seems to contradict Occam's razor (or ``law of parsimony''), we make it consciously
since we are not trying to prove the presence of spatial variation in the
$d_S$-covariance field. It is an established fact in socio-linguistics that there is
spatial variation in speech covariance  \citep[as can be seen in the vowel space
analyses of][]{clopper:2005,strange:2007,clopper:2008,fox2017reconceptualizing,renwick2017analyzing}.  We are therefore trying to estimate it in the
best possible way by choosing the most flexible model that fits the data, as long as
it is as good as less variable models. 
While there are possible reasons why the spatial variation in $d_S$-covariance is not
evident from the cross-validation curves, which could, for example, include the
confounding effect of sex or age on the MFCCs, the microphone and room reverberation
effect, or a small number of sound tokens where there is a mismatch between the
geographical location of the recording and the spoken dialect of that region, we
shall see in Section~\ref{sec:d-cov-field} that there is in fact evidence to support
that the $d_S$-covariance field is not constant.

% is
% slightly less noisy than the the minimum error is
% obtained with the nearest observation varying bandwidth for the mean and with the nearest location varying
% bandwidth for the $d_S$-covariance. The examination of the curves of average cross-validation error in
% Figures~\ref{fig:aast-mean-cv} and \ref{fig:aast-cov-sqrt-cv} leads to the choice of the smoothing parameters
% $\{h=1.5; k=300 \textrm{ nearest\ observations}\}$ for the mean and $\{h=2; k=20 \textrm{ nearest\
    % locations}\}$ for the $d_S$-covariance.

\begin{figure}[p]
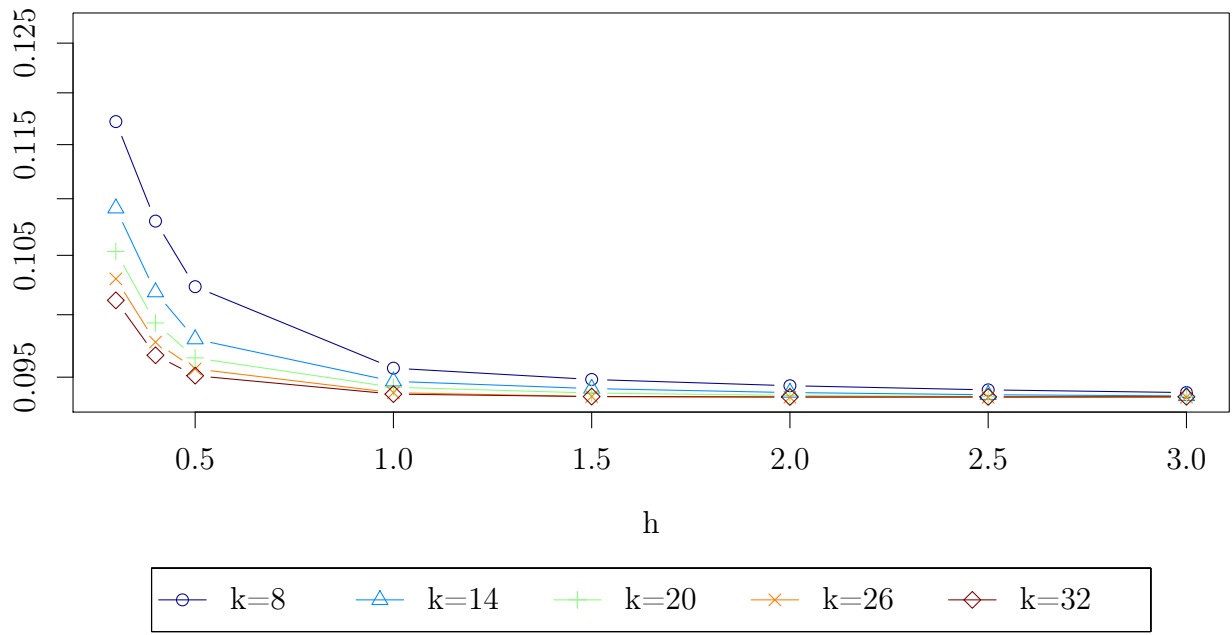

  \centering
   \input{./figures/cv-aast-mean-loc-new.tex}\\
  \input{./figures/cv-aast-cov-sqrt-loc-new.tex}\\
   \caption{Cross-validation curves of the ``class'' dataset for the mean MFCC field (top) and the
     $d_S$-covariance field (bottom)  when the bandwidth is adjusted using the $k$-th nearest locations.}
  \label{fig:aast-cv-loc}
\end{figure}

\subsection{Projection of the mean field onto principal components}
\label{sec:projection-pc}

Visualization of the field of mean MFCC is not a straightforward task. Indeed, at each location $x$ in
Great Britain, $\hat m(x)$ is an element of $L^2 \left( [0,1], \bR^p \right)$. A vizualization of the field $\hat m$
can be obtained by projection onto suitable elements of $L^2 \left( [0,1], \bR^p \right)$, i.e.\ by
looking at the map $x \mapsto \sc{ \hat m(x), \vfi }$ for various $\vfi \in L^2 \left( [0,1], \bR^p \right)$.
Here we choose to project onto the principal components of the MFCCs $\{Y_{lj}  : l = 1,\ldots, L ; j = 1,
\ldots, n_l \}$ (i.e.\  the pointwise multivariate PCA, that
is, the multivariate PCA of $Y_{jl}(t)$ evaluated over a discrete grid of values
$t$, or other words, our PCA is based on the eigen-analysis of the sample
covariance matrix of $Y_{lj}(t)$---and not its correlation matrix---evaluated
over a discrete grid of values $t$; another approach could be to use the method
proposed in \citet{chiou2014multivariate}).
This allows the reproduction of the geographical variation of the projections which
capture most of the variability in the original data and to compare the fields estimated for different values of $h$ and $k$, the projection directions being independent from them.

\begin{figure}[p]
  \centering
   \includegraphics[page=1, width=\linewidth]{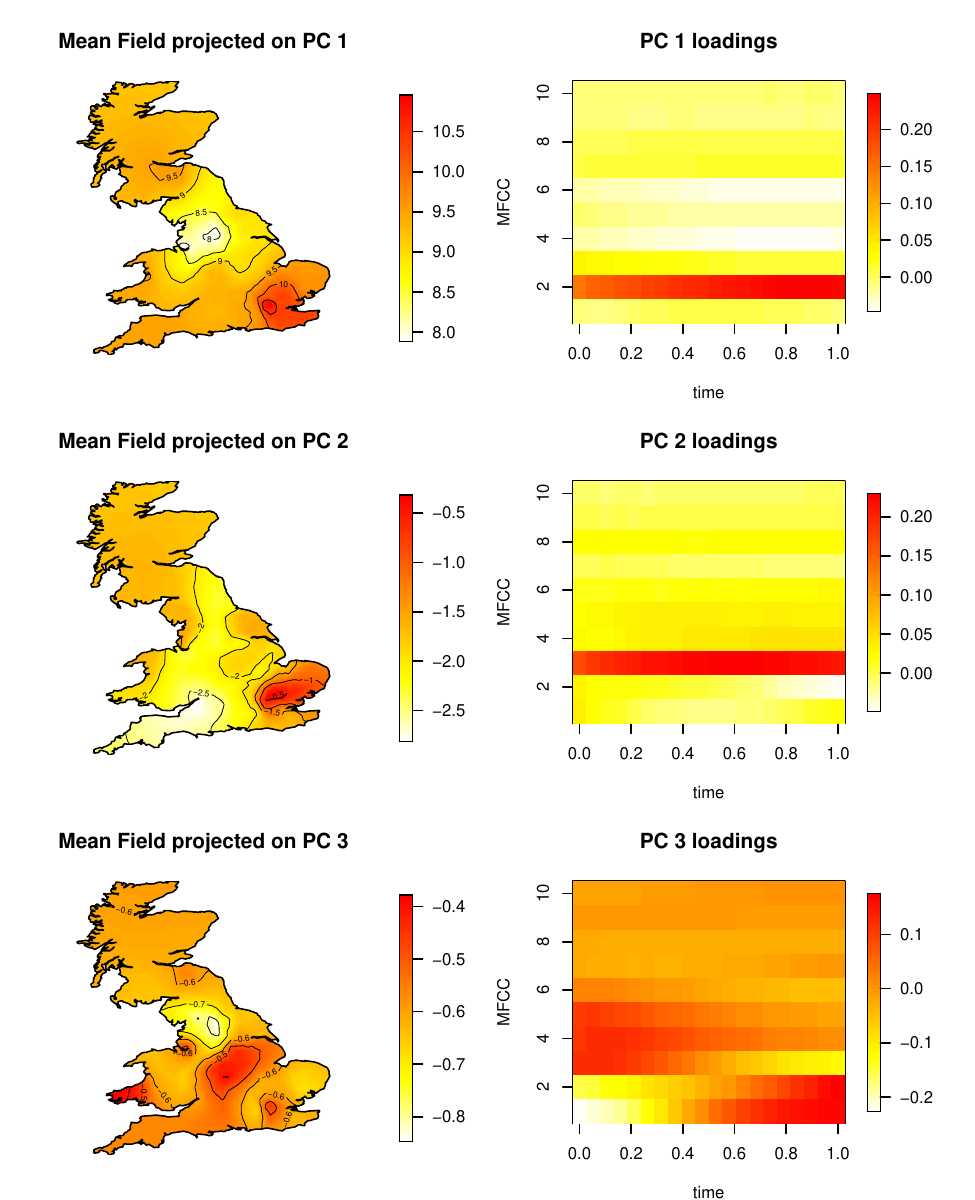}
  \caption{Left: Color maps with contours of the mean smooth MFCC field obtained for
    the ``class'' vowel with $h=0.5$ and $k=14$th nearest locations (denoted \emph{NL map} in the text), projected
    onto the first three principal components directions (from top to bottom) of the
    original data $\left\{Y_{lj}(t): l=1,\ldots, L; j=1,\ldots, n_l \right\}$.
  Right: Colour image representing the projection directions (loadings). }
  \label{fig:aast-mean-optim-nearest-loc}
\end{figure}

The maps of the projections of the estimated field for the choice of $h$ and $k$ that minimises the
cross-validation error can be found Figure~\ref{fig:aast-mean-optim-nearest-loc} (this corresponds to the nearest locations
bandwidth, with $h=0.5, \, k=14$). The maps of projections for the
 nearest observations bandwidth can be found in Figure~\ref{fig:aast-mean-optim-nearest-obs} of the
 \supplement.
 The first principal component direction (which accounts for $25.2\%$ of the total
 variance) essentially considers the energy
on the second cepstral coefficient ($92.4\%$ of its total energy), i.e.\ on the low
frequencies. The second principal component (which accounts for $19.6\%$ of the total
variance) essentially
considers the energy in the third cepstral coefficient ($91.2\%$ of its total energy), again energy in the low frequencies. The third
principal component direction (which accounts for $8.8\%$ of the total
variance) mainly consists of time dynamics (along the sound length) of the relative
volume, and the second and third cepstral coefficients, with some moderate time dynamics in the cepstral
coefficients $4$ and $5$. The fact that most of the energy in the first and second PC loadings concentrate on
a single cepstral coefficient ($92.4\%$ and $91.2\%$ of their respective total energy) confirms that the MFCC representation is indeed a suitable representation for
speech sounds. 
The map of the mean field
projected into the first principal component direction highlights
the difference between the region around London and the rest of the country,
in particular part of North England (and most strongly around Bradford).
The projection into the second principal component direction produces
high values in East England, and contrasts these values with South West, West Midlands, Yorkshire and
the Humber, and North East England, with the strongest contrast being with South West England.
The projection in the third principal component direction produces low values in North England, and contrasts
this region with isolated regions, such as East Midlands, the London area, and South West Wales.
% However, we have already seen that this direction is difficult to interpret phonetically and therefore it is hard to draw relevant conclusions.
Figures~\ref{fig:gb-counties} and \ref{fig:gb-regions} in the \supplement show the
counties and regions of Great Britain, and are provided as geographical
aids\footnote{The maps shown in this paper do not include the Isle of Wight, located
off the south coast of Great Britain, as there is no data present here, and as it is
not simply connected to the rest of mainland UK, it is not possible to provide smooth
estimates there.}.

% In order to check the results of the cross-validation analysis, we also consider in Figure~\ref{fig:aast-mean-less-smooth} the
% estimated field for the smaller value of $h=0.5$, again with $k=300$. While this less smooth estimate for the
% mean field shows a few local features that may be removed with the smoother estimate, overall it is clearly
% overfitting in the regions where more observations are available, such as Yorkshire or the South-East.
% Therefore, we conclude that the level of smoothing suggested by the cross-validation procedure is reasonable
% for the mean field estimate.

In order to assess whether there is spatial information in the mean field estimate, we compare our estimates
with a simulation where the mean and the error terms have no spatial information. The results of the simulation
provide evidence in support of spatial structure for the mean field, which is expected since the
cross-validation curves have a clear minimum. Details of the simulation are given in
Section~\ref{sec:simulation-study} of the \supplement.

\subsection{MFCC $d_S$-Covariance Field}
\label{sec:d-cov-field}

While the mean MFCC field captures the information about the average dialect sound changes, the regional
variability of such dialect sounds may well also be of considerable interest.  We therefore also want to
explore how the $d_S$-covariance changes over the region of interest. While it is in principle possible to use
dimension reduction methods, the interpretation of projections of the $d_S$-covariance may be
problematic, as discussed in Section~\ref{sec:d-cov}.  An alternative way to represent the $d_S$-covariance
variation is to consider a single location of interest and plot the square-root distances (averaged over the
length of the sound) between the $d_S$-covariance at the location of interest,  and the $d_S$-covariances at
all other locations of the map. This produces $2$D surfaces that reflect which parts of the country are more
similar or dissimilar to the location of interest with respect to $d_S$-covariance. However, such maps are not
directly interpretable, because many of their features appear due to the smoothing method. Indeed,
Figure~\ref{fig:contour-essex} shows the contours of the pairwise distances from Harlow (Essex), overlaid by
contours obtained from 100 simulations from a model with constant mean and constant $d_S$-covariance field
(details of the simulations are given in Section~\ref{sec:simulation-study} of the \supplement), a procedure which can be considered a bootstrap approximation to the underlying null field. We
can see in the Figure that the general form of the contours of the data and the simulations have similar
shapes (such systematic effects are not present for
the mean MFCC field, as can be seen from Figure~\ref{fig:global-model-mean} of the \supplement).
This is because the raw $d_S$-covariances $\breve \Omega_l$ are quite noisy (indeed,
Figure~\ref{fig:dcov-dist-vs-geographical-dist} in the
\supplement, which shows the pairwise distance between the raw $d_S$-covariances against their
corresponding geographical distance, has a nugget). Even though the contours of the data and the simulations
have similar shapes, there are some significant differences between them. In Figure~\ref{fig:contour-essex},
we notice that as one moves away from Harlow (Essex), the distances between the $d_S$-covariances are growing
faster in the data than what would be expected if there was no spatial structure in
the $d_S$-covariance field (this can be seen by noticing that the thick dashed lines
are not always contained in the bulk of the thin lines of the corresponding color).
However, this is not true for all regions of Great Britain. Indeed, Figure~\ref{fig:contour-lancashire}, which
shows the pairwise distances from Morecambe (Lancashire), does not exhibits such features as clearly. In
principle, one could look at such maps of contours of distances from each region of Great Britain to assess
whether or not the $d_S$-covariance field is varying spatially, but this is cumbersome and not visually
appealing.  A more appropriate tool for this purpose is to represent a normalized version $z_D(x,y)$ of the
pairwise distances between $d_S$-covariances at locations $x$ and $y$.  The definition of $z_D(x,y)$ is as
follows:
\begin{equation}
  z_D(x,y) = \frac{D(x,y) - \overline{D^*}(x,y)}{\sigma^*(x,y)},
  \label{eq:dist-zscore}
\end{equation}
where $D(x,y)$ is the distance between the $d_S$-covariances at $x$ and $y$ estimated from the data,
$\overline{D^*}(x,y)$, respectively $\sigma^*(x,y)$, is the average, respectively the standard deviation, of $\{
  D^{*b}(x,y): b=1,\ldots, 100\}$, where $D^{*b}(x,y)$  is the distance between the $d_S$-covariances at $x$ and $y$
for the $b$-th simulation replicate (for both the data and all the simulations, the smoothing parameters are
$h=1$, $k=32$ nearest locations). The notation $z_D(x,y)$ is chosen because \eqref{eq:dist-zscore} can be
interpreted as a z-score for the distance between the $d_S$-covariances between locations $x$ and $y$ of the
data, under the null hypothesis that the $d_S$-covariance field is constant.
Figure~\ref{fig:zscores-for-two-locations} shows the surfaces $\{ z_D(x_0,y): y \in \england
\}$ for the two locations $x_0 \in \england$ corresponding to the contours of
Figures~\ref{fig:contour-essex}~and~\ref{fig:contour-lancashire}. We can see in
Figure~\ref{fig:zscores-for-two-locations}
that the values of $y \in \england \mapsto z_D(x_0, y)$, for $x_0$ corresponding to Harlow (Essex), are all
larger than $2$ for $y$ in the Midlands and South of England, indicating a difference between
their $d_S$-covariance and that of Harlow.
 For $x_0$ corresponding to Morecambe (Lancashire), the
 values of $y \in \england \mapsto z_D(x_0, y)$ are below 2 (and even negative) for
 $y$ in
 North East England, East Midlands, and South East England, indicating little to no
 difference between their $d_S$-covariance and that of Morecambe. These conclusions are
 coherent with (and make more precise) those made from
 Figures~\ref{fig:contour-essex}~and~\ref{fig:contour-lancashire}.
 Instead of looking at each surface $\{ z_D(x_0, y): y \in \england\}$ separately,
it is possible to consider many such surfaces for locations all over Great Britain to get a global
appreciation of the variation of the $d_S$-covariance field.  Figure~\ref{fig:aast-zscores-map} shows the maps
associated to many representative locations together with their geographical position in Great Britain, a
``map of maps''.
We can see in the Figure that there is a very strong indication that the $d_S$-covariances of the region
around Glasgow and Edinburgh are different from those of North England, and
that the $d_S$-covariance of the Midlands are different from those in South and South-West England.
There is also very strong indication that the $d_S$-covariance around
Northamptonshire is different from
those of East England and South-East England, and moderate to strong indication that the $d_S$-covariances of
South England are different from those of the rest of England.
All of these interpretations should be of course tempered by the fact that they are
drawn from a very crude univariate representation of the $d_S$-covariance field (namely, the $z$-scores of
their pairwise distances), and while it allows to find regions where the $d_S$-covariance field is varying
spatially, it is not clear if a small value of the z-score $z_D(x,y)$ implies that there is no difference
between the $d_S$-covariances at $x$ and $y$.
Figures~\ref{fig:gb-counties} and
\ref{fig:gb-regions} in the \supplement show the counties and regions of Great Britain, and are
provided as geographical aids.

%%Figure~\ref{fig:aast-cov-map} shows the maps associated to many representative locations together with
%%their geographical position in Great Britain, a ``map of maps''. It is again possible to recognise the distinction
%%of the Greater London region as well as a divide between the South-East vs. the North of England, Wales and South-West. There
%%are however also some regions---including the contiguous regions of Northamptonshire, Cambridgeshire, Norfolk
%%and Suffolk and also the somewhat separate location of Hampshire---which show differences from both the
%%South-East and the North. Such a ``transitional'' or ``interface'' region is not captured well (and is in fact obscured) by the hard boundaries implied by conventional isogloss maps.

\begin{figure}[p]
  \begin{center}
    \includegraphics{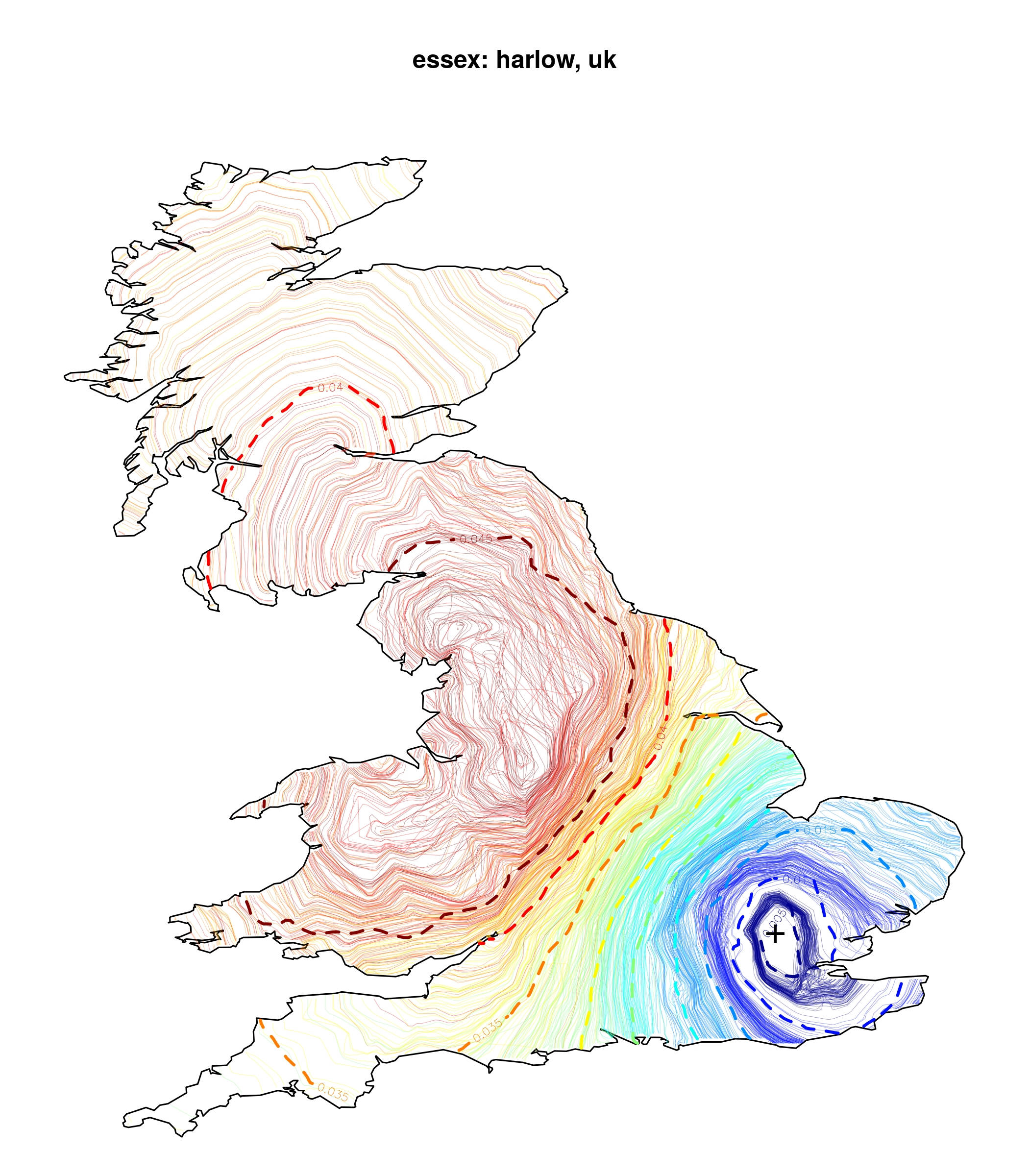}
  \end{center}
  \caption{Contours of the pairwise distances between the $d_S$-covariances at Harlow (Essex), and other
    location in Great Britain. The thick dashed lines correspond to the contours (level sets) for the BNC
    data, and the corresponding contours for each of the 100 simulations are given in thin lines, with the
    corresponding color.}
  \label{fig:contour-essex}
\end{figure}

\begin{figure}[p]
  \begin{center}
    \includegraphics{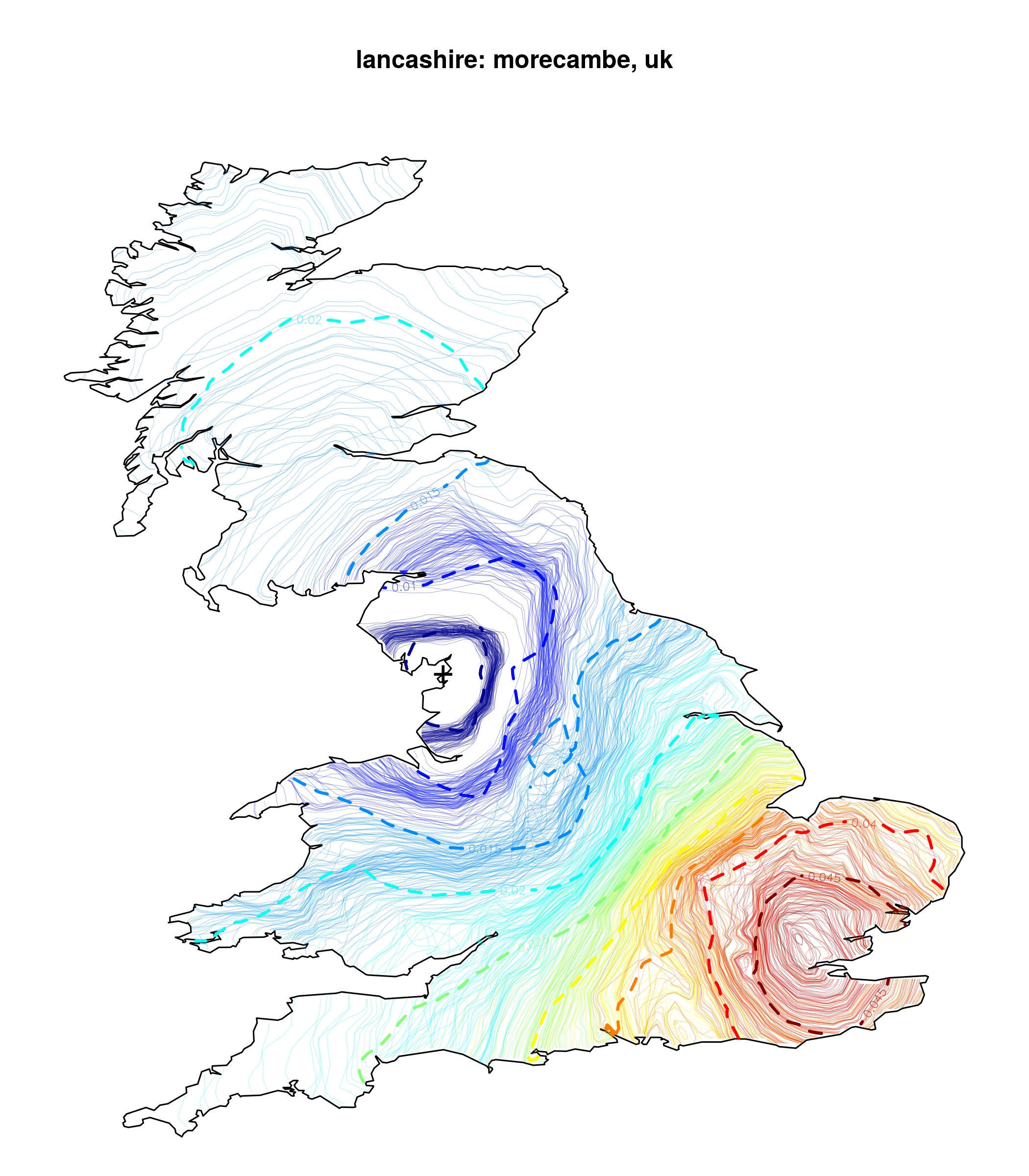}
  \end{center}
  \caption{Contours of the pairwise distances between the $d_S$-covariances at Morecambe (Lancashire), and other
    location in Great Britain. The thick dashed lines correspond to the contours for the BNC data,
    and the corresponding contours for each of the 100 simulations are given in thin lines, with the
    corresponding color.}
  \label{fig:contour-lancashire}
\end{figure}

\begin{figure}[h]
  \begin{center}
    \includegraphics{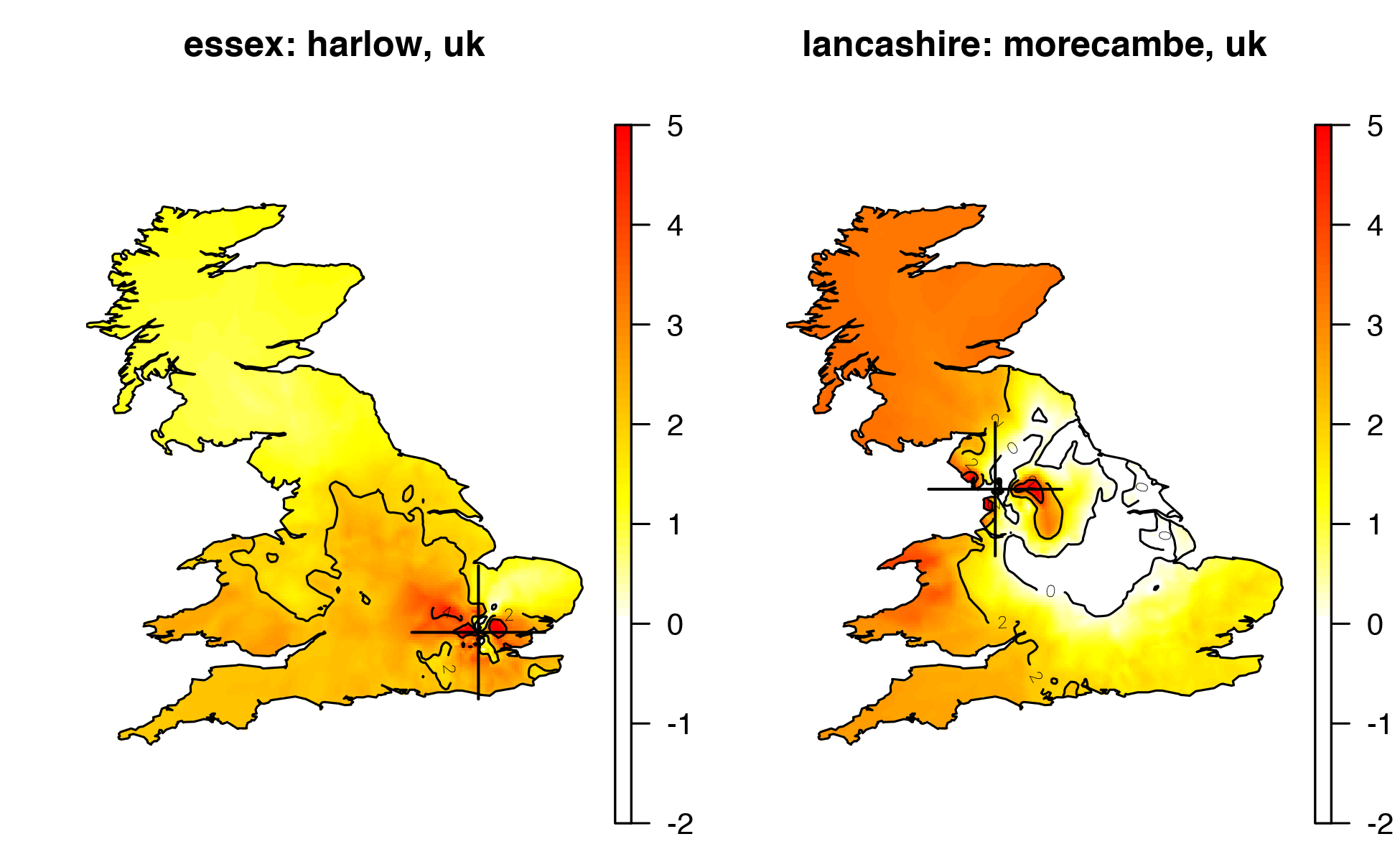}
  \end{center}
   \caption{Z-scores $\{ z_D(x_0,y) : y \in \england \}$ of the pairwise distances between the
     $d_S$-covariances from Harlow (Essex) (left), and from Morecambe (Lancashire) (right). The ``+'' on the maps represent the
     location of the corresponding $x_0$. Notice that the left sub-figure indicates
     a difference between the $d_S$-covariance of Harlow and those of the Midlands
     and South of England, whereas the right sub-figure indicates little to no
     difference between the
     $d_S$-covariance of Morecambe and those of North East England, East Midlands,
     and South East England. The definition of $z_D(x_0,y)$ is given in \eqref{eq:dist-zscore}. }
  \label{fig:zscores-for-two-locations}
\end{figure}

\begin{figure}[p]
  \centering
  \includegraphics{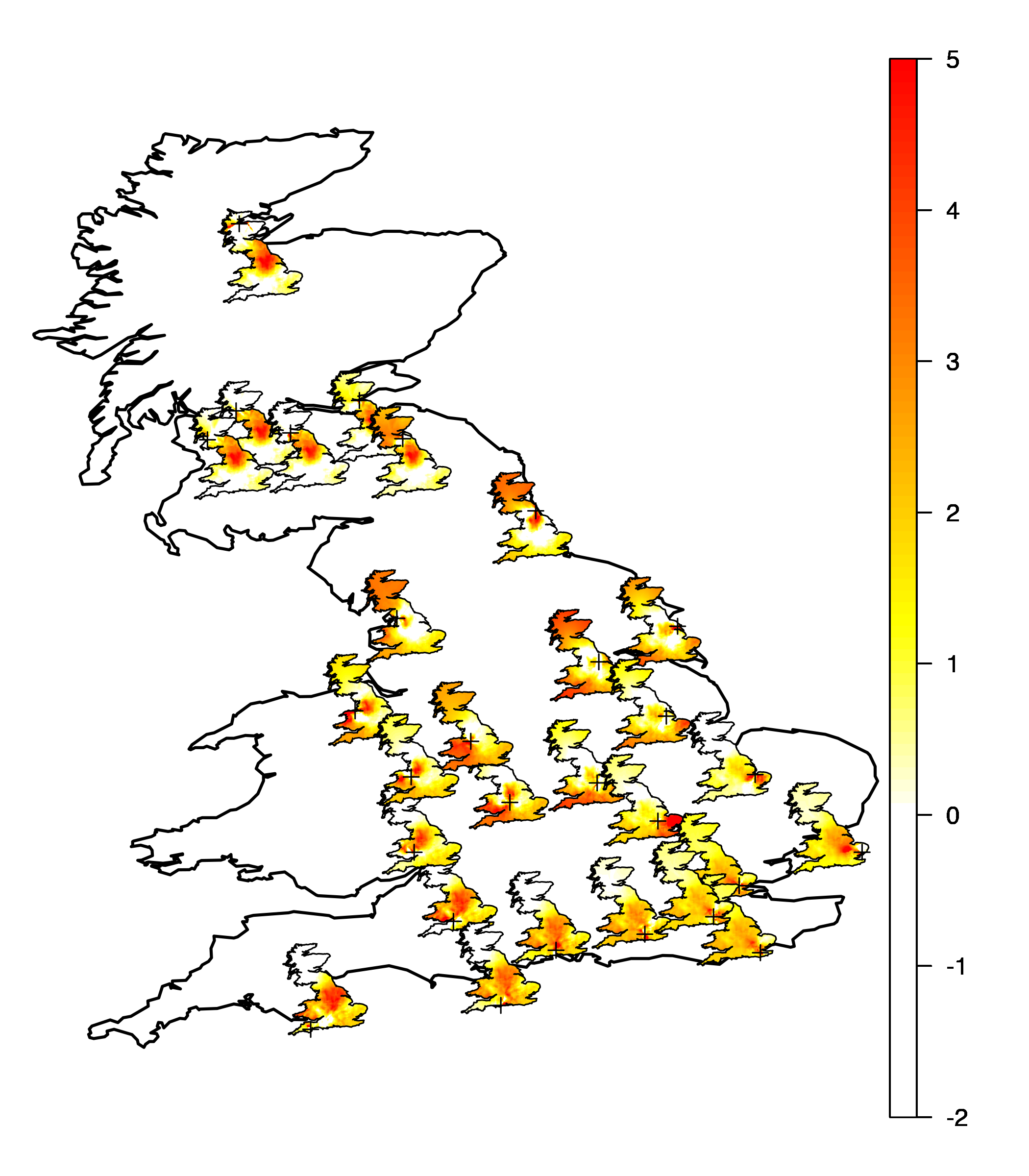}
   \caption{Z-scores $\{ z_D(x_0,y) : y \in \england \}$ of the pairwise distances between the
     $d_S$-covariances for a number of representative locations $x_0$ in Great Britain, reported on the
     corresponding positions on the geographical man of Great Britain. The '+' on the small maps represent the
     location of the corresponding $x_0$. The definition of $z_D(x_0,y)$ is given in \eqref{eq:dist-zscore}. }
  \label{fig:aast-zscores-map}
\end{figure}

\section{Discussion}
\label{sec:discussion}

We presented a method to explore spatial variation of sound processes which is of interest in particular for
dialectology and comparative linguistics. The need to model the change in the covariability between
frequencies, as well as in the mean sound, led us to propose the novel statistical concept of $d$-covariance,
i.e.\ a definition of covariability that relies on a metric distance $d$ different from the Euclidean (Frobenius)
distance. This allows the use of metrics that do not produce swelling effects, while
estimating the $d$-covariance consistently in the locations where observations are available. In particular,
we chose the square root distance $d_S$ described in \citet{dryden2009} because it is defined for positive
semi-definite matrices, and an explicit expression is available, as we showed in Section
\ref{sec:model-and-estimation}. It is clear that other metrics could be used within this framework,
and indeed recent work on choosing metrics \citep{petersen2016frechet} and smoothing under general metrics
\citep{pertersen:2016regression} could prove relevant to this setting.
However it is important to remember that the choice of metric should be considered within a data application
context as well.

We used a Mel Frequency Cepstral Coefficients (MFCC) representation for the sound data objects because this
has been found empirically to provide a better sound reconstruction, especially in the modified version of the
algorithm proposed by \citet{Erro2014}. Moreover, the fact that the frequency domain is partitioned into a
relatively small number of channels (through a weighted averaging over a range of contiguous frequencies) makes this representation more robust to small frequency misalignments
across speakers. MFCCs can be then treated as multivariate functional data, and we proposed a model where both
the mean and the $d_S$-covariance between coefficients change smoothly in space. We proposed to estimate these
smooth fields with a non-parametric estimator, and showed that this provides consistent estimates both for the
mean and for the $d_S$-covariance field. We also integrated into the smoothing procedure a geographical distance
based on the shortest path on the mesh used to triangulate the possibly non-convex region of interest. This
required a non-trivial argument to show the consistency of the derived estimator and it has a wider
applicability wherever there is the need of accounting for a complex geographical domain.

The proposed method allows, for the first time, the sound variation to be studied using speech recordings directly
(as opposed to phonetic transcription), and provides a continuous model for the sound change (through its
mean and $d_S$-covariance) in place of discrete regions boundaries, such as those traditionally reported in
isoglosses. We analysed speech data from the spoken part of the British National Corpus, and focused on the
pronunciation of the vowel in words such as \emph{fast} or \emph{class}, which is known to vary on a dialect
basis \citep{upton2013}, and has particularly prominent variations in British English. While it is possible
to listen to the reconstructed sounds (as given in the Supplementary Materials), visual maps are often useful
to recognise both global patterns and local features.  Exploring the estimated mean and $d_S$-covariance
fields, we uncovered geographical patterns that resemble established findings about the vowel pronunciation
(such as the contrast between the North and South-East England). However, the variation appears to be somewhat
smoother than expected (i.e.\  from traditional dialectological maps of 'isoglosses'), to the point where it
is possible to identify intermediate regions not easily classified by a hard clustering.
% 
% While the difference between the South West and the East/South East is well known, our dataset appears to show
% that the former presents some unique features that differentiate it from the Midlands and the North as well.
This invites additional studies to explore other sounds and further exploration of this and alternative
corpora.
Indeed, possible immediate extensions for this work include studying the joint behaviour of multiple
words/sounds in the language and taking into account additional (non geographical) covariates, such as
socio-economic variables.

\bigskip
\begin{center}
{\large\bf SUPPLEMENTARY MATERIAL}
\end{center}

Supplementary materials (functions used for the smoothing, sample sounds) can be obtained through the authors.

 % \clearpage

\bibliographystyle{agsm}
\citestyle{agsm}
\bibliography{structure/BNC-shahin,structure/BNC}

\clearpage

\appendices
\beginappendix

\section{Square Root of Symmetric Semi-positive Matrices}
\label{sec:square-roots}

We give here some useful properties of square root of matrices.
The following result states that square root of a symmetric  positive semi-definite matrix is unique.
\begin{thm}[e.g.\ \citet{axler:2015}]
  % \label{thm:uniqueness-of-square-root}
  Let $A$ be a $p \times p$ real matrix. If $A$ is symmetric  positive semi-definite, i.e. $A = A^\tp$ and
  \( x^\tp A x \geq 0, \forall x \in \bR^p,\)
  then there exists a unique positive $p \times p$ matrix $B$ such that $A = BB$.
  The matrix $B$ is called the square root of $A$, and is denoted by $\sqrt{A}$ or $A^{1/2}$.
\end{thm}
In particular, this tells us that the square root distance between symmetric positive semi-definite matrices is well defined.
The following gives a explicit formula for the square root of symmetric positive semi-definite rank one matrices. It's proof follows from
direct calculations.
\begin{prop}
  Let $x \in \bR^p, x \neq 0$ and $A$ be a $n \times p$ matrix. Then
  \begin{enumerate}
    \item $(xx^\tp)^{1/2} = xx^\tp/\rpnorm{x}$
    \item $(Axx^\tp A^\tp)^{1/2} = Ax x^\tp A^\tp/\rpnorm{Ax}$ provided $Ax \neq 0$.
  \end{enumerate}
\end{prop}

\section{Technical results and proofs}

\label{sec:technical-results}

\begin{lma}
  \label{lma:soln-to-ds-cov-fit-criterion}
  The minimizer $\hat \Omega(x, \cdot) \in L^2\left( [0,1],  \symmat_p \right)$ of the following fit criterion:
  \begin{equation}
    \sum_{l = 1}^L  K_h\left( \graphdist(x, X_l) \right)
    \int_0^1 d_S^2(\breve \Omega_l(t), \hat \Omega(x, t)) dt,
    \label{eq:fit-criterion-cov-lemma}
  \end{equation}
  is given by
  \begin{equation}
    \hat \Omega(x,t) = \left[ \sum_{l = 1}^L w_l(x) \sqrt{\breve \Omega_l(t) }\right]^2,
    \label{eq:cov-est-inlemma}
  \end{equation}

  \begin{proof}
    Setting $\tilde w_l = K_h\left( \graphdist(x,X_l) \right)$, using the definition of $d_S$ and permuting the sum and
    integral in \eqref{eq:fit-criterion-cov-lemma}, we can rewrite
    \eqref{eq:cov-est-inlemma} as
    \[ \int_0^1 \sum_{l} \tilde w_l \HSnorm{ \sqrt{\breve \Omega_l(t)} - \sqrt{\hat
    \Omega(x,t)} }^2 dt. \]
    This expression is minimized with respect to $\hat \Omega$ by minimizing it for
    each $t$.  
    Fixing $t$ and writing $y_l = \sqrt{\breve \Omega_l(t)}$ and $y =  \sqrt{\hat
    \Omega(x,t)}$, and omitting the integral, the fit criterion becomes
    \[ \sum_l \tilde w_l \HSnorm{y_l - y}^2. \]
    This is just a weighted least-squares problem, whose solution is $y = \sum_l \tilde
    w_l y_l / \left( \sum_l \tilde w_l \right).$ Substituting $y, y_l$ back concludes
    the proof.
  \end{proof}
\end{lma}

\begin{proof}[Proof of Proposition~\ref{prop:d-S-cov-well-defined-and-explicit-formula}]
  Without loss of generality, assume that $\ee X = 0$.
  By definition, we have
  \begin{align*}
    \cov_{d_S}(X) &= \argmin_{\Omega \in \symmat_p} \ee d_S^2\left( XX^\tp, \Omega \right)
    \\ &= \argmin_{\Omega \in \symmat_p} \ee \HSnorm{  \sqrt{XX^\tp} - \sqrt{\Omega}
  }^2
    \\ &= \argmin_{\Omega \in \symmat_p} \ee d_E^2\left(\sqrt{XX^\tp}, \sqrt{\Omega}
  \right).
  \end{align*}
  The minimum is achieved for $\sqrt{\Omega} = \ee \sqrt{XX^\tp}$, hence
  $\cov_{d_S}(X) = \left( \ee \sqrt{XX^\tp} \right)^2 = \eee{\frac{XX^\tp}{\rpnorm{X}}} ^2$.
\end{proof}

\begin{proof}[Proof of Proposition~\ref{prop:d-covariance-and-linear-transformations}]
  Without loss of generality, assume that $\ee X = 0$.
  By~\eqref{eq:ds-cov-as-regularized-cov} of the paper, we
  have $\cov_{d_S}(AX) = \eee{\frac{AXX^\tp A^\tp}{\rpnorm{AX}}}^2 = A \eee{\frac{XX^\tp}{\rpnorm{AX}}}
  A^{\tp} A \eee{\frac{XX^\tp}{\rpnorm{AX}}} A^{\tp}$.
  The proof is completed by showing that
  $\cov_{d_{S,A}}(X)$ must satisfy
  $\cov_{d_{S,A}}(X) = A \eee{\frac{XX^\tp}{\rpnorm{AX}}} A^{\tp} A \eee{\frac{XX^\tp}{\rpnorm{AX}}}
  A^{\tp}$, which follows from an argument similar to the proof of
  Proposition~\ref{prop:d-S-cov-well-defined-and-explicit-formula}. 
  For the special case where $A^\tp A = I$, then $\rpnorm{AX} = \rpnorm{X}$ and 
  \[ \cov_{d_S}(AX) = A \eee{\frac{XX^\tp}{\rpnorm{X}}}^2 A^\tp = A \cov_{d_S}(X)
    A^\tp. \]
\end{proof}

  \begin{proof}[Proof of Proposition~\ref{prop:rate-of-convergence-sample-ds-cov}]
    Let $\tilde S = \left( \frac{1}{n} \sum_{i = 1}^n \sqrt{ (Y_i -
    \mu)(Y_i - \mu)^\tp } \right)^2$, and $S = \cov_{d_S}(Y)$. Recall that
    $S = \left( \ee \sqrt{ (Y - \mu)(Y - \mu)^\tp} \right)^2$. Let $\phi_x(y) =
    \sqrt{ (x-y)(x-y)^\tp }$, for $x, y \in \bR^p$. Notice that $\phi_x(y) =
    (x-y)(x-y)^\tp/\rpnorm{x-y}$ if $y \neq x$, and $\phi_x(x) = 0$. Furthermore, it is
    not difficult to show that $\phi_x$ is Lipschitz, i.e.\ $\HSnorm{\phi_x(y) -
    \phi_x(y')} \leq \kappa_p \rpnorm{y - y'}$, where $\kappa_p \geq 0$ does not
    depend on the value of $x$, but only on the dimension $p$.
    We therefore have
    \begin{align*}
      d_S( \hat S,  \tilde S ) &\leq \frac{1}{n} \sum_{i = 1}^n
      \HSnorm{\phi_{Y_i}(\overline Y) - \phi_{Y_i}(\mu)}
      \\ &\leq \frac{1}{n} \sum_{i = 1}^n \kappa_p \rpnorm{\overline Y - \mu}
      \\ &= \kappa_p \rpnorm{\overline Y - \mu},
    \end{align*}
    and $d_S(\hat S, \tilde S) = \bigop(n^{-1/2})$.

    The proof is completed by showing that $d_S(\tilde S, S) =
    \bigop(n^{-1/2})$, which follows from the central limit theorem applied to the random
    element $\sqrt{ (Y - \mu)(Y - \mu)^\tp }$.  The central
    limit theorem is indeed applicable here since
    \begin{align*}
      \eee{ \HSnorm{ \sqrt{(Y - \mu)(Y - \mu)^\tp} }^2} &= \eee{ \HSnorm{
  \frac{(Y - \mu)(Y - \mu)^\tp}{\rpnorm{Y - \mu}}}^2 }
  \\ &= \ee \rpnorm{ Y - \mu}^2 < \infty.
  \qedhere
\end{align*}
\end{proof}

  \begin{proof}[Proof of Theorem~\ref{thm:consistency}]
    By the triangle inequality,
    \begin{equation}
      \label{eq:eq1-consistency-sample-covd-field}
      d_S(\hat \Omega(x, t), \Omega(x, t)) \leq  d_S(\hat \Omega(x,t), \tilde \Omega(x,t)) +
      d_S(\tilde \Omega(x,t), \Omega(x,t)),
    \end{equation}
    where $\tilde \Omega(x,t )$ is the same as $\hat \Omega(x,t)$, but with the sample mean at the
    observations replaced by the true mean, i.e.\ $\tilde \Omega(x,t ) = \left( \sum_{l = 1}^L w_l(x) \sqrt{\tilde
        \Omega_l(t)}
    \right)^2$,
    \[
      \sqrt{\tilde \Omega_l(t)} = n_l^{-1} \sum_{j = 1}^{n_l} \sqrt{(Y_{lj}(t) -
        m_l(t))(Y_{lj}(t) -
        m_l(t))^\tp} = n_l^{-1} \sum_{j = 1}^{n_l} \sqrt{\vep_{lj}(t) \vep_{lj}(t)^\tp}
    \]
    and $m_l(\cdot) = \ee Y_{l1}(\cdot)$.

    Let us first look at the first term in \eqref{eq:eq1-consistency-sample-covd-field}.
    Writing $\ee_X$ for the expectation conditional on $X_1, \ldots, X_L$,
    the triangle inequality and H\"older's inequality yield
    \begin{align*}
      \ee_X d_S(\hat \Omega(x,t), \tilde \Omega(x,t))  & \leq  \sum_{l = 1}^L w_l(x) \sqrt{\ee_X d_S^2(\breve
        \Omega_l(t), \tilde \Omega_l(t) )}.
    \end{align*}
    By arguments in the proof of
    Proposition~\ref{prop:rate-of-convergence-sample-ds-cov} of the paper, we have
    \[
      \ee_X d_S^2( \breve \Omega_l(t), \tilde \Omega_l(t) ) \leq \kappa_p^2 n_l^{-1} \ee
      \rpnorm{\vep(x,t)}^2\bigg|_{x=X_l} \leq \kappa_p^2 c^{-1} n^{-1} \sup_{x \in \england, t \in [0,1]} \ee
      \rpnorm{\vep(x,t)}^2,
    \]
    which is non-random, and independent of $t$. Since $\sum_{l} w_l(x) = 1$, we get
    \begin{align*}
      \ee_X d_S(\hat \Omega(x,t), \tilde \Omega(x,t))  & \leq
      \frac{\kappa_p}{ \sqrt{c n} } \sup_{x \in \england, t \in [0,1]} \sqrt{\ee \rpnorm{\vep(x,t)}^2}
    \end{align*}

    Let us now look at the term $\ee_X d_S(\tilde \Omega(x,t), \Omega(x,t)) \leq \sqrt{ \ee_X d_S^2(\tilde
      \Omega(x,t), \Omega(x,t))}$. Since
    \[
      \ee_X d_S^2( \tilde \Omega(x,t), \Omega(x,t)) = \sum_{r,s=1}^p  \ee_X \left( \left[
          \sqrt{\tilde \Omega(x,t)}
        \right]_{rs} - \left[ \sqrt{\Omega(x,t)} \right]_{rs} \right)^2,
    \]
    it is enough to control the mean square error of each coordinate of $\sqrt{ \tilde \Omega(x,t) }$.
    Notice that $\sqrt{\tilde \Omega(x,t)}  = \sum_{l = 1}^L w_l(x)  \sqrt{\tilde \Omega_l(t)} $,
    Therefore we can apply Lemma~\ref{lma:smoothing-mse-univariate} to each coordinate $1 \leq r \leq s \leq p$
    (by symmetry), with $Z_l(t) = \left[ \sqrt{\tilde \Omega_l(t)} \right]_{rs}$.
    Since
    $\ee_X \sqrt{\tilde \Omega_l(t)}  =  \sqrt{\Omega(X_l, t)}$
    and
    \begin{align*}
      \varX{ \left[ \sqrt{\tilde \Omega_l(t)} \right]_{rs} } &= n_l^{-1} \varX{ \frac{\vep(X_l,t)_r \vep(X_l,
          t)_s)}{\rpnorm{\vep(X_l, t)} } }
      \\ &\leq n_l^{-1} \ee_X \rpnorm{\vep(X_l, t)}^2
      \\ &\leq \frac{1}{c n} \sup_{x \in \england, t \in [0,1]} \ee \rpnorm{\vep(x,t)}^2
    \end{align*}
    the Lemma can be applied with $m(x,t) = \sqrt{\Omega(x,t)}$ and $\hnorm{\nu}_\infty \leq (c n )^{-1} \sup_{x \in
      \england, t \in [0,1]} \ee \rpnorm{\vep(x,t)}^2$. For fixed $r,s$, the conditional squared bias is bounded by
    $ \bigop(h^2) $
    and the conditional variance term is bounded by
    $ \bigop\left( \frac{1}{n L h^2} \right),$ both bounds being uniform in $t$. The proof is finished by combining these last results.
  \end{proof}

\begin{lma}
  \label{lma:smoothing-mse-univariate}
  Assume $(X_l, Z_l(t)) \in \england \times \LLR, l=1, \ldots, L$ are i.i.d., with $X_l \simiid f$, and assume
  $Z_l | X_l$ are i.i.d. with mean $\eee{Z_l(t) | X_l} = m(X_l, t)$ and $\var{ Z_l(t)| X_l  } = \nu(X_l, t)$,
  and that Conditions~\ref{cond:kernel-and-graph-distance}, \ref{cond:density} of
  the paper hold.
 Furthermore, assume
  \begin{enumerate}
    \item For each $t \in [0,1]$, $m(\cdot, t) : \england  \rightarrow \bR$ is $C^1$, and
      \[
        \hnorm{ \nabla_x m }_\infty := \sup_{x \in \england, t \in [0,1]} \rpnorm{\frac{\partial m}{\partial x}(x,t)} < \infty,
      \]
    \item $f$ is a continuous density on $\england$,
    \item $\hnorm{\nu}_\infty := \sup_{x \in \england, t \in [0,1]} \nu(x,t)  < \infty$ for each $x \in \england$.
  \end{enumerate}
  Let $\hat m(x,t) = \sum_{l = 1}^L w_l(x)Z_l(t)$, where $w_l(x)$ is defined in
  \eqref{eq:nadaraya-watson-weights-cov-field} in the paper.
  Then for each $x$ in the interior of $\england$, if $f(x) > 0$, we have
  \begin{equation}
    \label{eq:bias-technical-lemma}
    |\ee_X \hat m(x,t) - m(x,t) | \leq \frac{2\pi \mu_2(K) \hnorm{f}_\infty \hnorm{\nabla_x m}_\infty
      c_2^2}{c_1^2 f(x)}\left[ h + \littleop(h) \right],
  \end{equation}
  and
  \begin{equation}
    \label{eq:var-technical-lemma}
    \varX{\hat m(x,t)} \leq \frac{\hnorm{\nu}_\infty}{Lh^2} \left[  \frac{c_2^4}{c_1^4 f^2(x)} + \littleop(1)\right]
  \end{equation}
  as $L \rainf, h \raz$ such that $Lh^2 \rainf$, where the remainder terms are uniform in $t$.
  \begin{proof}
    Without loss of generality,
    assume that $K$ is renormalized such that $\int_0^\infty K(s) s ds = (2\pi)^{-1}$,
    and let $\tilde K_h: \bR^2 \rightarrow [0,\infty)$ be defined by
    $\tilde K_h(x) =  K( \rpnorm{x} / h )/h^2 = K_h(\rpnorm{x})$ for $h > 0$. Notice that $\tilde K_h$ is a
    valid density function on $\bR^2$ for any $h > 0$, and that it is an approximate identity as $h \raz$.

    We first give a technical result that will be useful, and whose proof follows from standard arguments: for any $\alpha, \beta \geq 0$,
    \begin{equation}
      \int_\england K_h^{1+ \alpha}\left(  \rpnorm{x-y}  \right) \rpnorm{x-y}^\beta f(y) dy \leq 2\pi \mu_{\beta +1}(K)
      \hnorm{K}_\infty^\alpha \hnorm{f}_\infty \cdot h^{\beta - 2\alpha}.
      \label{eq:bound-multivariate-moments-of-kernel}
    \end{equation}
    Recall that $\hat m(x,t) = \left[ \sum_{l = 1}^L K_h(\graphdist(x,X_l)) \right]^{-1} \sum_{l = 1}^L
    K_h(\graphdist(x,X_l))Z_l(t)$.
    First, notice that
    \[
      L^{-1} \left[ \sum_{l = 1}^L K_h(\graphdist(x,X_l)) \right] = \int_{\england} K_h( \graphdist(x,y) ) f(y)
      dy + \bigop \left( \left[ L^{-1} \int_{\england} K_h^2 (\graphdist(x,y)) f(y) dy \right]^{1/2} \right).
    \]
    By Condition~\ref{cond:kernel-and-graph-distance}  of the paper and
    \eqref{eq:bound-multivariate-moments-of-kernel}, the
    stochastic term is of order $\bigop(1/\sqrt{L h^2})$.
    Concerning the integral, since $K_h$ is an approximate identity as $h \raz$, approximation theory gives
    \begin{align*}
      \int_{\england} K_h( \graphdist(x,y))f(y) dy \geq c_2^{-2} \int_{\england} \tilde K_{h/c_2}( x-y )f(y)dy
      = c_2^{-2} f(x) + \littleo(1)
    \end{align*}
    as $h \raz$. Therefore, as $h \raz, L \rainf$,
    \begin{equation}
      \label{eq:bound-denominator}
      \left[L^{-1} \sum_{l = 1}^L K_h(\graphdist(x,X_l)) \right]^{-1} \leq \frac{c_2^2}{f(x)} + \littleop(1).
    \end{equation}
    Let us now look at the bias term. First, notice that $\ee_X \hat m(x,t) = \sum_{l = 1}^L w_l(x) m(X_l,
    t)$.
    Since $x \mapsto m(\cdot, t)$ is $C^1$, for all $x,y \in \england$, Taylor's theorem yields $m(y, t) =
    m(x,t) + r(x,y,t)$, where $|r(x,y,t)| \leq \hnorm{\nabla_x m}_\infty \rpnorm{x-y}$.  Therefore, using
    \eqref{eq:bound-denominator},
    \begin{align*}
      | \ee_X \hat m(x,t) - m(x,t)| &\leq
      \left[ \frac{c_2^2}{f(x)} + \littleop(1) \right] \hnorm{\nabla_x m}_\infty  \cdot \left[ L^{-1}\sum_{l=1}^L
            K_h( \graphdist(x,X_l) ) \rpnorm{x-X_l} \right]
    \end{align*}
    The second term in square brackets is now approximated:
    \begin{align*}
      L^{-1}\sum_{l=1}^L K_h( \graphdist(x,X_l) ) \rpnorm{x-X_l} &\leq c_1^{-2} L^{-1}\sum_{l=1}^L K_{h/c_1}(
      \rpnorm{x- X_l} ) \rpnorm{x-X_l}
 \\ &=
  c_1^{-2} \int_{\england} \tilde K_{h/c_1}( x - y ) \rpnorm{x-y} f(y)dy
  \\ & \qquad + c_1^{-2}\bigop \left( \left[ L^{-1} \int_{\england} K_{h/c_1}^2 ( \rpnorm{x-y}) \rpnorm{x-y}^2 f(y) dy \right]^{1/2} \right)
  \\ & \leq c_1^{-2} 2\pi \mu_{2}(K) \hnorm{f}_\infty \cdot h + \bigop(1/\sqrt{L}).
 % \\ & \leq h c_1^{-2} f(x) + \littleo(h) + \bigop(\sqrt{1/L})
 % \\ & \leq h c_1^{-2} 2 \pi \mu_2(K) f(x) + \littleop(h)
    \end{align*}
    Combining these results with \eqref{eq:bound-denominator} yields the conditional bias term \eqref{eq:bias-technical-lemma}.

    Concerning the variance, we have
    \begin{align*}
      \varX{\hat m(x,t)} &= \left[ L^{-1 } \sum_{l = 1}^L K_h(\graphdist(x,X_l)) \right]^{-2} \cdot
      L^{-1}\left[
        \frac{1}{L} \sum_{l = 1}^L K_h^2(\graphdist(x,X_l))\nu(X_l,t)\right]
      \\  &\leq \left[ \frac{c_2^4}{f^2(x)} + \littleop(1) \right] \cdot
      \hnorm{\nu}_\infty L^{-1} c_1^{-4}
      \left[ \frac{1}{L} \sum_{l = 1}^L K_{h/c_1}^2(\rpnorm{x-X_l})\right],
    \end{align*}
    Where we have used \eqref{eq:bound-denominator}. For the term in the second square brackets, we have
    \begin{align*}
      \left[ \frac{1}{L} \sum_{l = 1}^L K_h^2( \rpnorm{x-X_l})\right] &= \int_\england K_h^2(
      \rpnorm{x-y})f(y)dy + \bigop\left( \left[ \frac{1}{L} \int_{\england} K_h^4( \rpnorm{x-y}) f(y)
          dy \right]^{1/2} \right)
      \\ &\leq \hnorm{K}_\infty \hnorm{f}_\infty h^{-2}  + \bigop(1/\sqrt{L h^6}),
    \end{align*}
    where we have used \eqref{eq:bound-multivariate-moments-of-kernel}. Combining these results yields the
    conditional variance bound \eqref{eq:var-technical-lemma}.
  \end{proof}
\end{lma}
The following Lemma gives the approximation error in using the sample total variance in place of the true
variance in the estimator of the mean field \eqref{eq:estimated-mean-field} in the
paper. Let $\sigma^2(x) = \ee
\hnorm{\vep(x)}^2$.
\begin{lma}
  Assume
  \begin{equation}
    \label{eq:vep-norm-4-bounded}
    \sup_{x \in \england} \ee \hnorm{\vep(x)}^4 < \infty,
  \end{equation}
  \begin{equation}
    \label{eq:obs-per-loc-asymptotically-of-order-n}
    c' n \leq n_l \leq C'n, \:l=1,\ldots, L, \quad \text{ for some constants } c', C', \text{ as } n \rainf,
  \end{equation}
  \begin{equation}
    \label{eq:inf-total-variance-positive}
    \inf_{x \in \england} \sigma^2(x) > 0 \qquad \& \qquad \sup_{x \in \england} \sigma^2(x) < \infty.
  \end{equation}
  Let $\hat m$ be defined as in \eqref{eq:estimated-mean-field} and let $\check{m}(x) = \sum_{l=1}^L
  \lambda_l(x) \overline Y_l$, where
  \[
    \lambda_l(x) = \tilde \lambda_l(x)/ \sum_{l=1}^L \tilde \lambda_l(x) \qquad \& \qquad  \tilde \lambda_l(x) =
    n_l K_h( \graphdist(x,X_l) )/\sigma^2(X_l).
  \]
  Then, for fixed $L, h$,
  \[
    \left| \hat m(x, t) - \check m(x, t) \right| \leq O_p(n^{-1/2}) \max_{l = 1, \ldots, L} |\overline
    Y_l(t)|, \quad \text{as } n \rainf.
  \]
  \begin{proof}
    First, notice that $| \hat m(x, t) - \check m(x, t) | \leq \max_{l = 1, \ldots, L} |\overline Y_l(t)|
    \left| \sum_{l} w_l(x) - \lambda_l(x) \right|$. For the rest of the proof, will drop the $x$ to
    simplify notation, and write $w_l$ instead of $w_l(x)$.
    Notice that
    \begin{equation*}
    \left| \sum_{l} w_l - \lambda_l \right| \leq
    % \leq \frac{\sum_l \left[ |\tilde \lambda_l - \tilde w_l|  s_{\tilde w} + \tilde w_l  |
        % s_{\tilde w} - s_{\tilde \lambda}| \right]}{s_{\tilde \lambda} s_{\tilde w}},
    \frac{\sum_l  |\tilde \lambda_l - \tilde w_l|  }{s_{\tilde \lambda}} +
    \frac{ | s_{\tilde w} - s_{\tilde \lambda}| }{s_{\tilde \lambda} },
    \end{equation*}
    where $s_{\tilde w} = \sum_{l} \tilde w_l$ and $s_{\tilde \lambda} = \sum_{l} \tilde \lambda_l$.
    Using \eqref{eq:obs-per-loc-asymptotically-of-order-n}, we get that
    \begin{equation}
    \left| \sum_{l} w_l - \lambda_l \right| \leq \left( C'/c' \right)^2 \frac{\sum_l  |\breve \lambda_l - \breve w_l|  }{s_{\breve \lambda} } +
    \left( C'/c' \right)^2 \frac{ |
        s_{\breve w} - s_{\breve \lambda}| }{s_{\breve \lambda} },
      \label{eq:bound-on-weight-difference}
    \end{equation}
    where the `` $\breve \cdot$ '' entries are the same as the `` $\tilde \cdot$ '' entries, but without the
    $n_l$s, i.e.\ $\breve w_l = K_h( \graphdist(x,X_l) )/ \hat \sigma^2(X_l)$, and $\breve \lambda_l = K_h(
    \graphdist(x,X_l) )/ \sigma^2(X_l)$.
    Using \eqref{eq:vep-norm-4-bounded} and the delta method, we have
    \[
      \breve \lambda_l - \breve w_l = K_h( \graphdist(x,X_l) ) \cdot \bigop(n^{-1/2}).
    \]
    The first summand in \eqref{eq:bound-on-weight-difference} is now bounded:
    \begin{align*}
    \frac{\sum_l |\breve \lambda_l - \breve w_l|   }{s_{\breve \lambda} }
        & \leq \frac{\bigop(n^{-1/2}) \sum_l K_h( \graphdist(x,X_l) )}{ \sum_l K_h( \graphdist(x, X_l)
      )/\sigma^2(X_l) }
    \\ &= \bigop(n^{-1/2}),
    \end{align*}
    where we have used \eqref{eq:inf-total-variance-positive}.
Using the same arguments, we get the same bound on the second summand of \eqref{eq:bound-on-weight-difference},
    % notice that \eqref{eq:inf-total-variance-positive} and
    \begin{align*}
    \frac{| s_{\breve w} - s_{\breve \lambda}| }{s_{\breve \lambda}}
    & \leq  \bigop(n^{-1/2}).
    \end{align*}
    The proof is finished by combining these results.
  \end{proof}
\end{lma}

\section{Preprocesing of the British National Corpus Data} \label{sec:data}
 We describe here in further detail the preprocessing of the sound data extracted from the spoken part of the British National Corpus and analyzed in the paper. 

  \subsection{Raw Data Preprocessing}
  First all the segmentation information and all the contextual
  information were  extracted. Then, the list of words for the segmentation and the context were corrected for coding
  differences (e.g.\ ``they'll'' was coded as two separate words ``they'' and ``'ll'' in the contextual
  information files). After this, the segmentation and contextual information were merged together. This was
  done by matching---within each audio recording file---consecutive groups of words. The algorithm we used
  looked for a unique sequence of words of length $L$ that perfectly matched between the two sets of words.
  The algorithm looped through the sequence of utterances (sequence of words pronounce by the same speaker)
  defined in the contextual XML files, by initially setting $L$ to the minimum of the length of the utterance
  and $50$ (this was chosen for speeding up the matching). If multiple matches were found, $L$ was increased and the search was performed
  again. If no match was found, $L$ was decreased and the
  search was performed again. If the algorithm didn't find any match, or if $L > 50$, the algorithm went to
  the next word in the current utterance (setting $L = 1$). Then $L$ was either increased, respectively
  decreased,  if multiple matches, respectively no match, was found. If $L > 50$, the algorithm was restarted
  with $L=1$ but the perfect matching was relaxed to approximate matching using the
  \emph{optimal string aligment} metric \citep{VanderLoo2014}, with distance at most 2.

  The result of the preprocessing is a data frame with variables \texttt{word}, \texttt{begintime},
  \texttt{endtime}, \texttt{textgridfilename}, \texttt{index}, \texttt{agegroup}, \texttt{role}, \texttt{sex},
  \texttt{soc}, \texttt{dialecttag}, \texttt{age}, \texttt{persname}, \texttt{occupation}, \texttt{dialect},
  \texttt{id}, \texttt{placename}, \texttt{activity}, \texttt{locale}, \texttt{wavfile},
  \texttt{placenamecleaned} and about 5 million observations (i.e.\ words).
  Discriminative information about the speaker is missing for about 2.9\% of the words, and information about
  the location of the recording is missing for about 8.4\% of the words.

  %\subsection{Subsetting the Data}
%
%  For the purposes of the current paper, we restricted ourselves to the analysis of sounds of the vowel ``a''
%  present in the
%  following list of words:
%  \begin{equation}
%    %\texttt{bath, path, pass, class, glass, grass, past, last, brass, blast, ask, cast, fast}.
%    \texttt{pass, class, glass, grass, past, last, brass, blast, ask, cast, fast}.
%    \label{eq:list-of-words}
%  \end{equation}
%  This list was established by linguistic knowledge that the vowels in these words are pronounced in the same
%  (geographically consistent) way and therefore we can consider them as the replicates of the same sound. We
%  denote this as the ``class'' dataset.
%
%%  The vowel ``a'' in these words is expected to exhibit the
%%  same patterns of variation across England. Initially, there were about 10800 sounds in the dataset
%%  corresponding to the words in \eqref{eq:list-of-words}.

  \subsection{Cleaning}

  Since the data we analyzed are sounds from noisy recording, we first cleaned the sounds corresponding to the
  set of words 
   \begin{equation}
      %\texttt{bath, path, pass, class, glass, grass, past, last, brass, blast, ask, cast, fast}.
      \texttt{class, glass, grass, past, last, brass, blast, ask, cast, fast,pass}.
      \label{eq:list-of-words-appendix}
    \end{equation}
  The following sounds were removed:
  \begin{enumerate}
    \item Sounds with duration outside the interval $[0.2, 1]$ seconds.
    \item 400 sounds with the lowest maximal amplitudes.
    \item Sounds corresponding to young speakers (selected by taking speakers less than 10 year old and whose
      median pitch was above a fixed threshold)
  \end{enumerate}
  To further remove low quality sounds from our analysis, we ranked the sounds $s_1, \ldots, s_N$, for each
  word $w$ in \eqref{eq:list-of-words-appendix}, according to following score,
  \begin{equation}
    \text{score}_i = \frac{1}{L_i} \sum_{l=1}^{L_i} \left( \check s_i(t_l) - \indic{[a(w), b(w)]}(t_l/t_{L_i}) \right)^2 \exp\left( - \indic{[a(w),
        b(w)]}(t_l/t_{L_i}) \right) ,
    \label{eq:score-cleaning}
  \end{equation}
  where $\check s_i(t_l) = \tilde s_i(t_l) / \max_{l=1,\ldots, L_i} \tilde s_i(t_l)$, $\tilde s_i$ is the
  root mean square amplitude (RMSA) of $s_i$ on a running window of $10$ milliseconds, and $a(w), b(w) \in [0,1]$
  were chosen by looking at the plot of $\tilde s_i$ for a sound of good quality, and correspond roughly to
  the location of the vowel in the sound. Large values
  of $\text{score}_i$ correspond to noisier sounds. The effect of
  the exponential factor in \eqref{eq:score-cleaning} is to give higher score to sounds having large RMSA
  outside the vowel interval, while still penalizing for low RMSA inside the vowel interval.
  For each word $w$ of our list of words, we then discarded the sounds with the largest $5\%$ scores.

  \subsection{Vowel Segmentation and MFCC Extraction}

  We extracted the MFCCs of all the sounds
  corresponding to the words in \eqref{eq:list-of-words-appendix}, using the software \texttt{ahocoder}
  (\url{http://aholab.ehu.es/ahocoder/index.html}) with parameter \verb+--CCORD=30 --LFRAME=16+.

  In order to extract  the MFCC corresponding to the vowel segment of the recording of the words in
  \eqref{eq:list-of-words-appendix}, we performed the following steps. For each word
  in \eqref{eq:list-of-words-appendix}:
  \begin{enumerate}
    \item align the MFCCs of the sounds of the word with respect to the first MFCC coefficient,
    \item find the segment of the warped sounds which corresponds to the vowel,
    \item extract the corresponding portion on the unwarped MFCCs,
    \item recompute all the unwarped  MFCCs on a common grid,
  \end{enumerate}

  \subsection{MFCC alignment}

  Let us describe more precisely the alignment step in the preprocessing procedure. Let $\mfcc{i}(t,m),  i = 1, \ldots, N$ denote the
  MFCCs of the sounds corresponding to the current word $w$. Recall that $m = 1, \ldots, M$, and assume that the time domains have been linearly
  rescaled, i.e.\ $t \in [0,1]$. We first align the curves $\mfcc{i}(\cdot, 1), i=1,
  \ldots, N$ using the Fisher-Rao metric. This yields warping functions $\gamma_i: [0,1] \rightarrow
  [0,1]$ such that $\mfcc{i}(\gamma_i(\cdot), 1), i=1, \ldots, N$ are aligned. Then we align all the MFCC
  coefficients of the sound $i$ using the warping $\gamma_i$, that is, we set $\widetilde{\mfcc{i}}(t,m) = \mfcc{i}(\gamma_i(t), m), t \in [0,1], m = 1,\ldots, M$ for all $i$.
  The idea is that, after alignment, the
  temporal location of the vowel would be the same accross all registered MFCCs of a same word, which would
  make the vowel segmentation much easier.

  Once the MFCCs corresponding to a common word $w$ have been aligned, the interval $[a(w), b(w)] \subset
  [0,1]$ corresponding to the vowel sound was found by manual auditory discrimination. The inverse of the
  warping functions
  were then used to compute the interval $I_i  = [\gamma_i^{-1}(a(w)), \gamma_i^{-1}(b(w))]$, which is the
  vowel interval of the $i$-th unaligned MFCCs. The interval $I_i$ was then linearly rescaled to $[0,1]$,
  yielding the vowel MFCCs
  \begin{equation}
    \mfcc{i}^\text{vowel}(t, m) = \mfcc{i}\Big( (1-t)\gamma_i^{-1}(a(w)) + t \gamma_i^{-1}(b(w)), m \Big), \quad t \in
    [0,1]
    \label{eq:vowel-mfcc}
  \end{equation}

\begin{figure}[p]
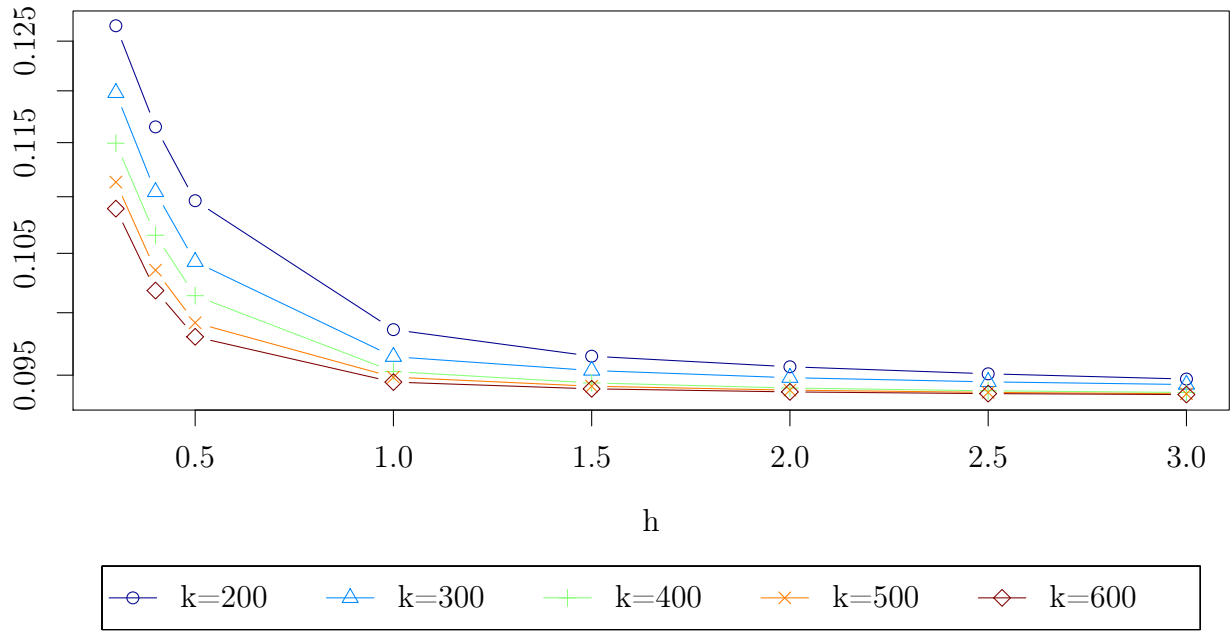

  \centering
   \input{./figures/cv-aast-mean-effective-new.tex}\\
  \input{./figures/cv-aast-cov-sqrt-effective-new.tex}\\
   \caption{Cross-validation curves of the ``class'' dataset for the mean MFCC field (top) and the
     $d_S$-covariance field (bottom)  when the bandwidth is adjusted using the $k$-th nearest observations.}
  \label{fig:aast-cv-obs}
  \thispagestyle{empty}
\end{figure}

\begin{figure}[p]
  \centering
   \includegraphics[page=1, width=\linewidth]{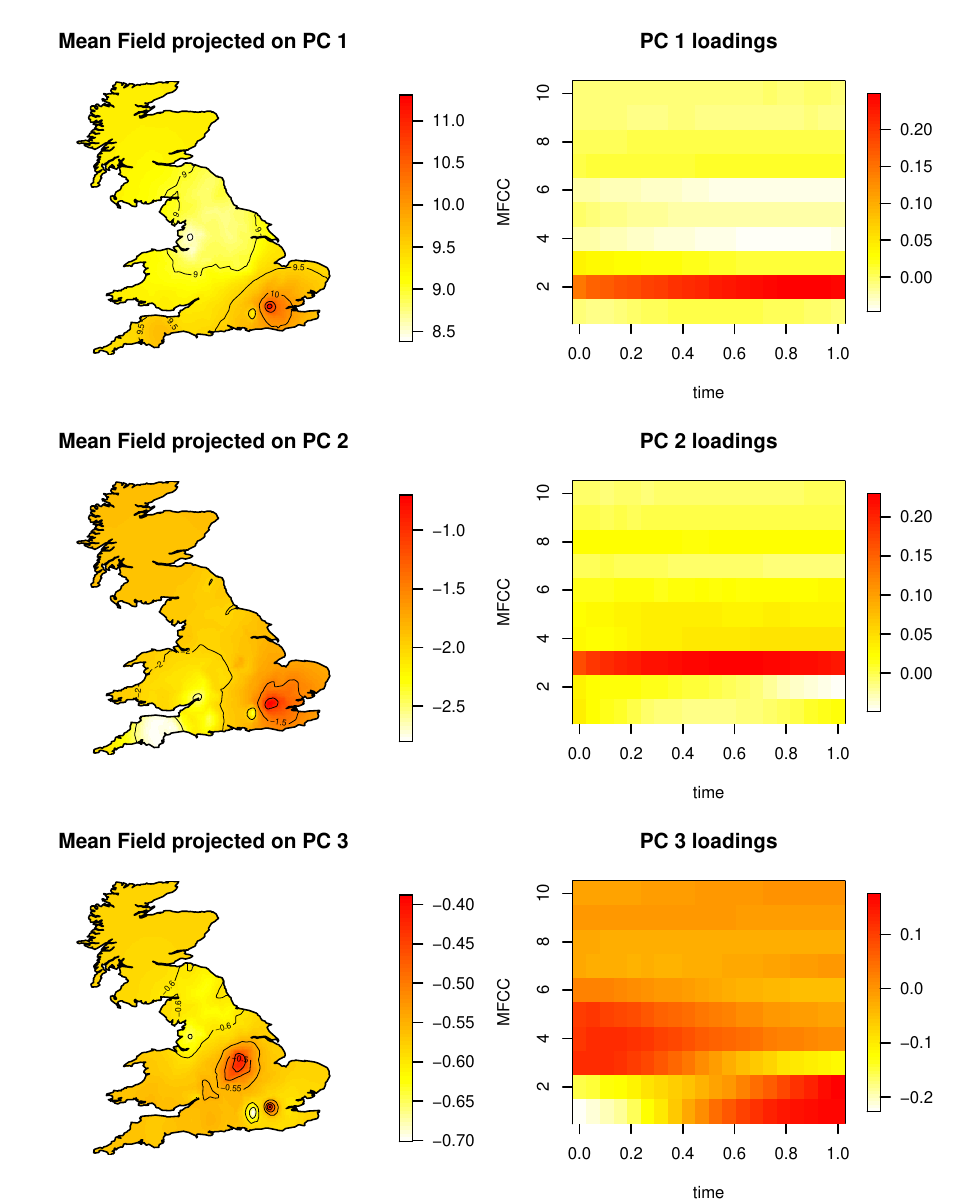}
  \caption{Left: Color maps with contours of the mean smooth MFCC field obtained for
    the ``class'' vowel with $h=1.5$ and $k=300$th nearest observations (denoted \emph{NO map} in the text), projected
    onto the first three principal components directions (from top to bottom) of the
    original data $\left\{Y_{lj}(t): l=1,\ldots, L; j=1,\ldots, n_l \right\}$.
  Right: Colour image representing the projection directions (loadings). }
  \label{fig:aast-mean-optim-nearest-obs}
\end{figure}

\begin{figure}[p]
  \vspace{-1cm}
  \begin{center}
    \hspace*{-1cm}\includegraphics[height=23cm]{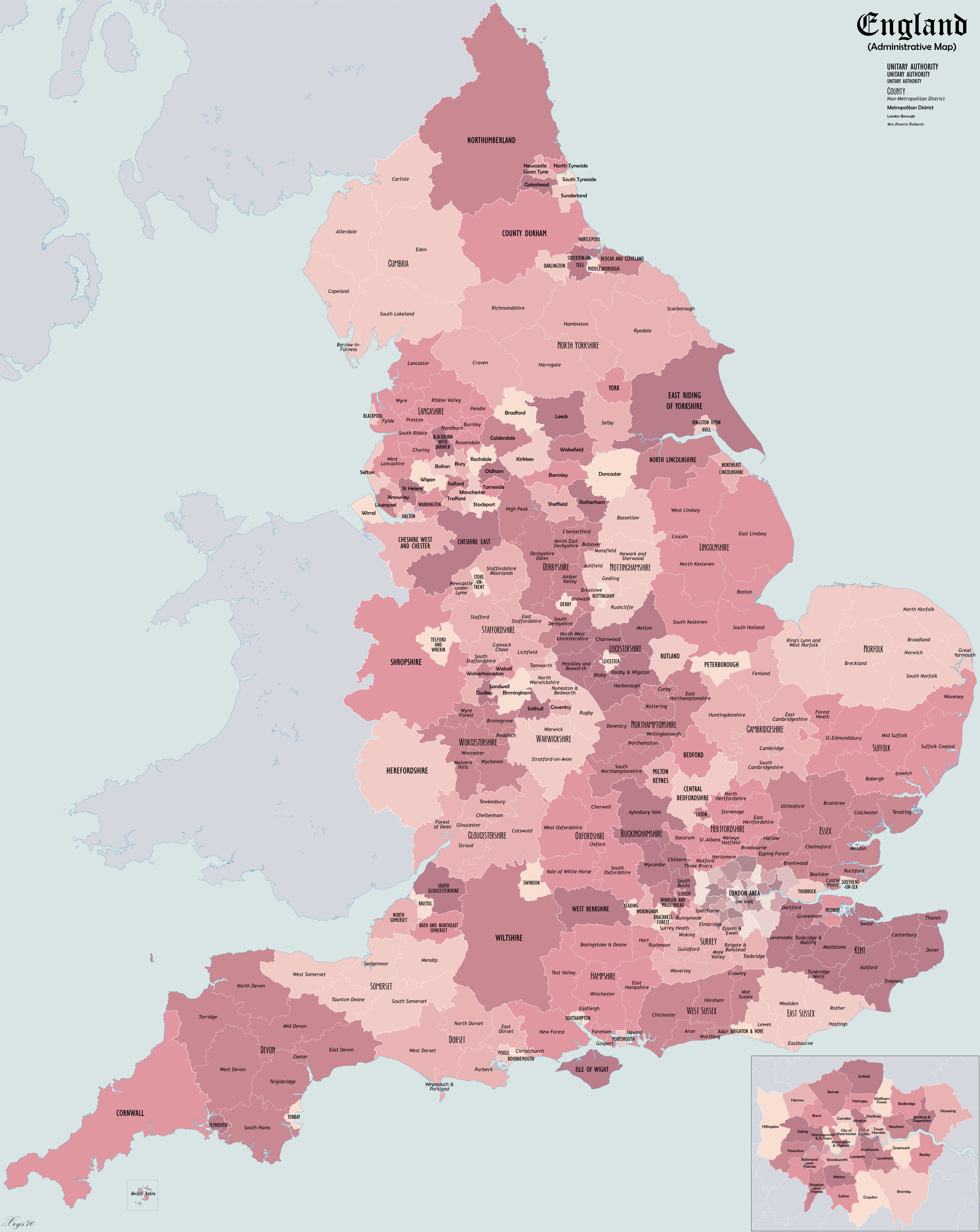}
  \end{center}
  \caption{Counties of England.
Licenced under the Creative Commons Attribution 3.0 Unported license. Attribution: XrysD.
\\ \url{https://en.wikipedia.org/wiki/File:England_Administrative_2010.png}.
}
  \label{fig:gb-counties}
\end{figure}

\begin{figure}[p]
  \begin{center}
    \includegraphics[height=18cm]{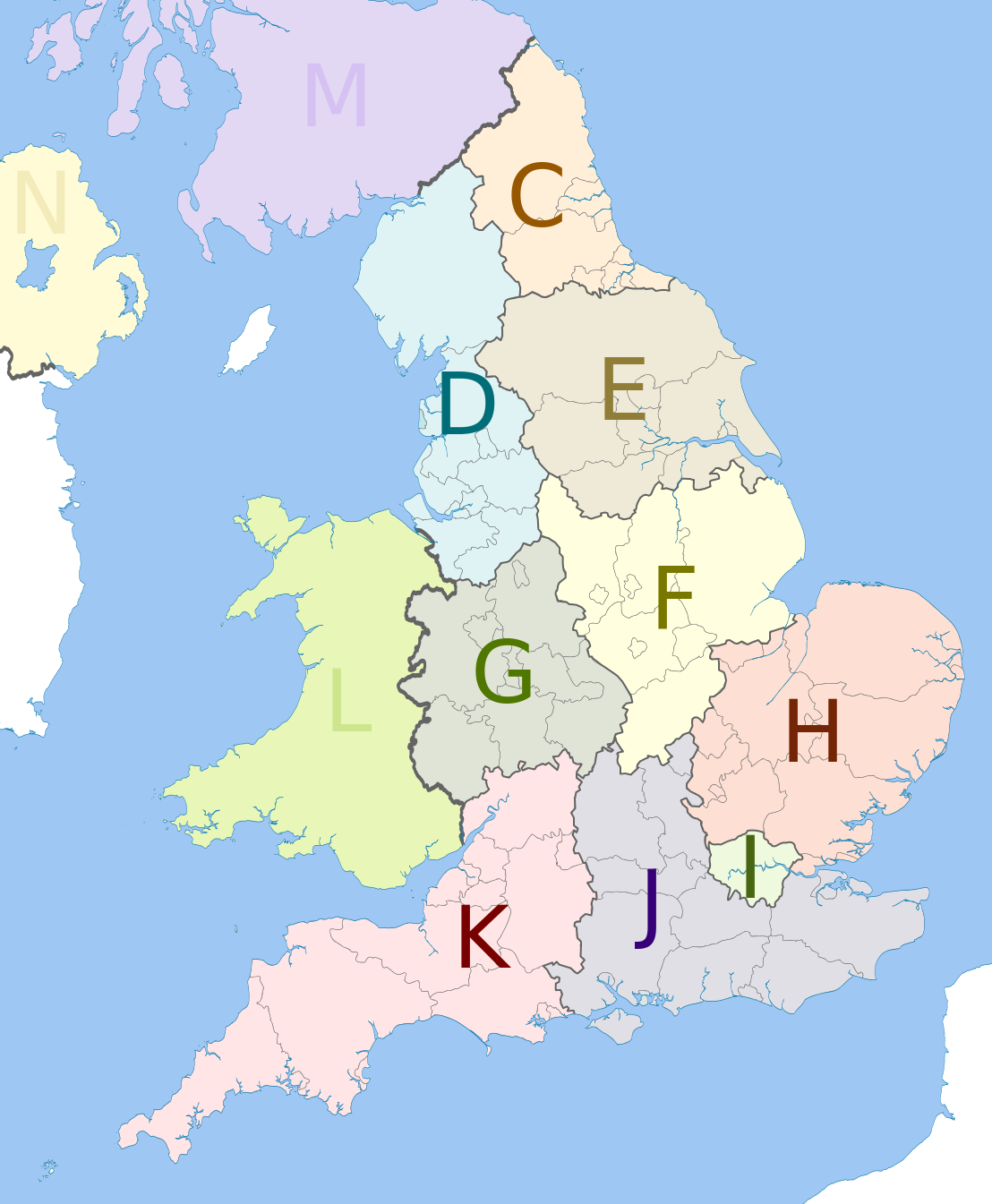}
  \end{center}
  \caption{Regions of Great Britain. C = North East England,
D = North West England,
E = Yorkshire and the Humber,
F = East Midlands,
G = West Midlands,
H = East of England,
I = Greater London,
J = South East England,
K = South West England,
L = Wales,
M = Scotland.
Licenced under the Creative Commons Attribution-Share Alike 3.0 Unported license. Attribution: Dr Greg and
Nilfanion. \\ \url{https://commons.wikimedia.org/wiki/File:NUTS_1_statistical_regions_of_England_map.svg}.
}
  \label{fig:gb-regions}
\end{figure}

\section{Modeling the Vowel Sound Duration}

The sound duration of the vowel in the words of the ``class'' dataset are believed to carry
part of the information of the spatial variation of the dialect sounds. However, since the duration
cannot capture time dynamics in relative volume, and differences in the vowel quality, the information
carried by the vowel duration is a very crude approximation of the vowel sound.  This is why the
focus of the paper is on the MFCCs of the vowel sounds.  We have nevertheless produced a spatial map  of
the relative duration of the vowel sound (relative to the duration of the word), where the spatial map
is obtained by spatial smoothing of the relative durations at each observation location, obtained using
a linear mixed model with \texttt{observation location}, \texttt{word} and \texttt{sex} as fixed
effects, and \texttt{speaker} as random effect. 
  The resulting map is given in
Figure~\ref{fig:sound-duration}, together with the projection of the mean MFCC field
onto the second
principal component.  The same spatial smoothing parameters have been used for both maps ($h=0.5, k=14$
nearest locations).  It can be seen that the two maps are quite correlated (the absolute correlation
is $0.66$; note that the principal component is defined up to a sign), and therefore the duration information is more or less similar to that obtained by the
projection of the MFCC mean field onto the second principal component.  

\begin{figure}[h]
  \begin{center}
    \includegraphics[width=\textwidth]{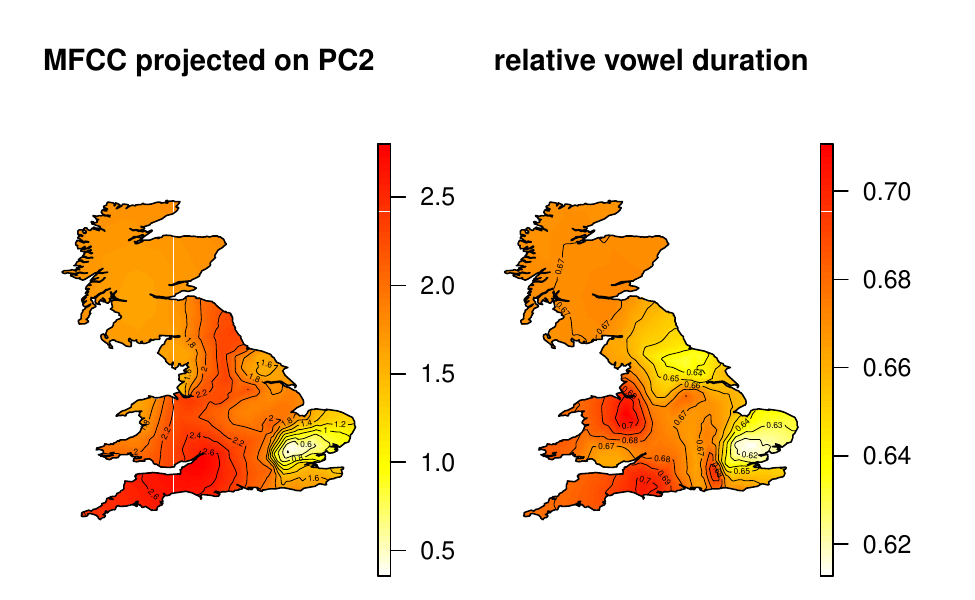}
  \end{center}
  \caption{Mean MFCC field projected on PC2 (left) and duration field of the vowel
  sounds (right). The absolute correlation between the two fields is $0.66$.}
  \label{fig:sound-duration}
\end{figure}

\begin{figure}[h]
  \begin{center}
    \includegraphics[width=\textwidth]{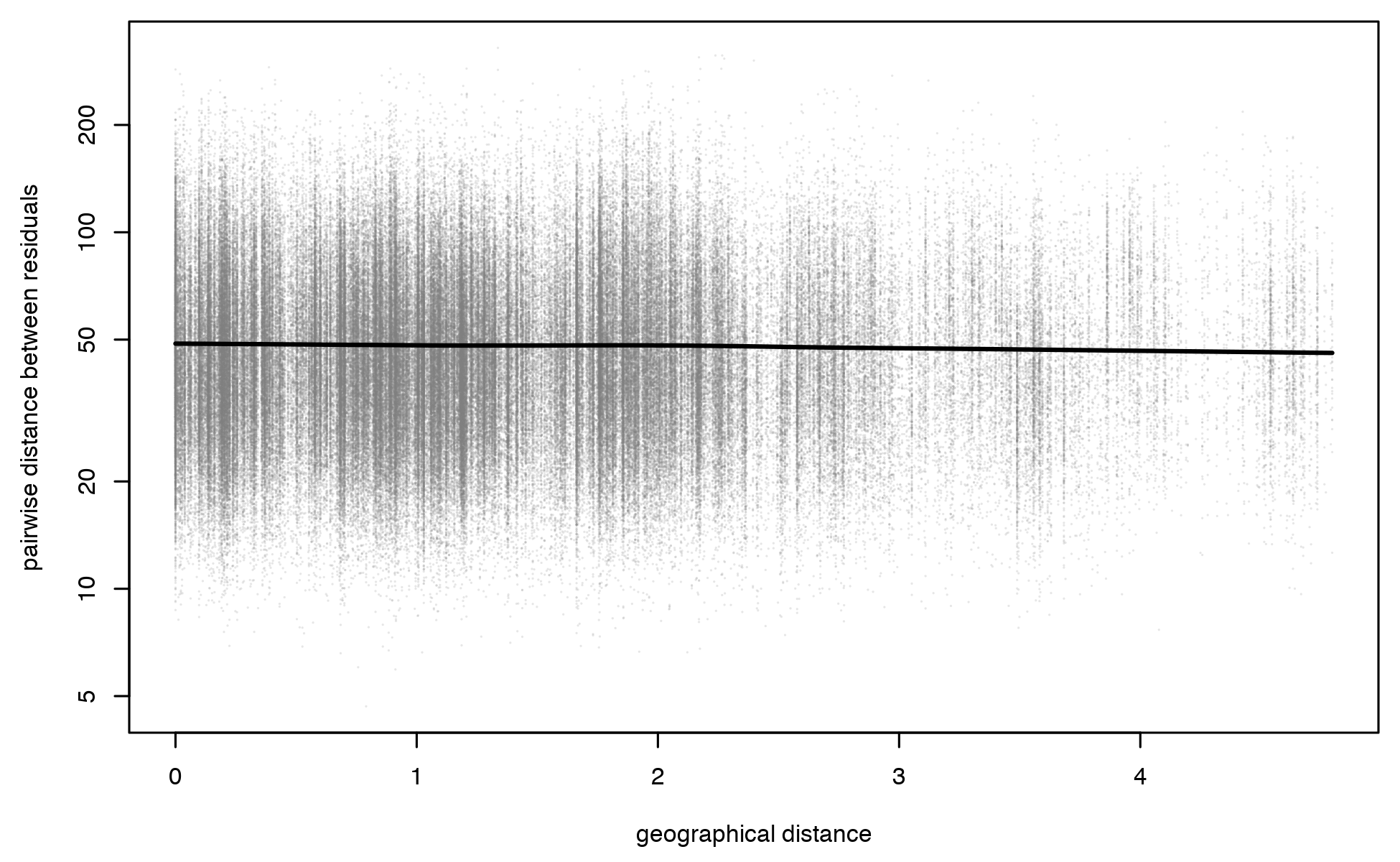}
  \end{center}
  \caption{Scatterplot of pairwise distance between residuals (y axis) against their geographical distance (x
    axis). The black thick line is a robust local linear regression obtained via the
  R function \texttt{lowess}.}
  \label{fig:scatterplot}
\end{figure}

\begin{figure}[h]
  \begin{center}
    \includegraphics[width=\textwidth, trim=0 200 0 200, clip=true]{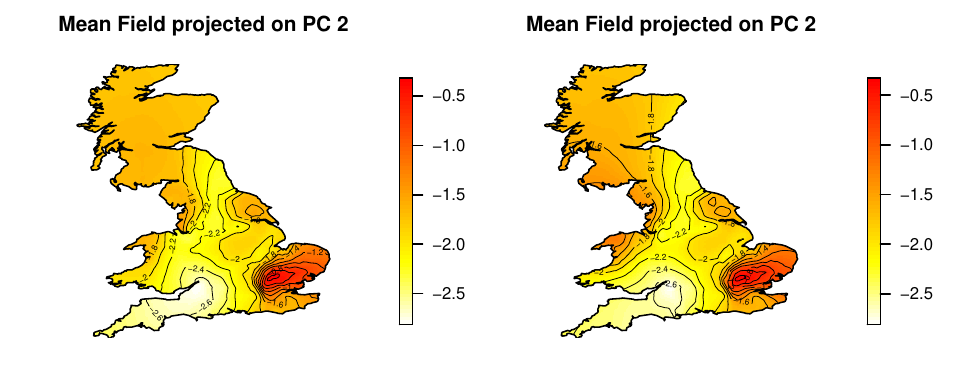}
  \end{center}
  \caption{Mean field of the BNC dataset projected on PC2 with computed with geodesic distance (left) and with
    Euclidean distance (right).  Notice the artifacts near the boundaries (the level curves go across the port
    of Edinburgh when using the Euclidean metric).
  }
  \label{fig:geodesic-vs-euclidean}
\end{figure}

\begin{figure}[h]
  \begin{center}
    \includegraphics{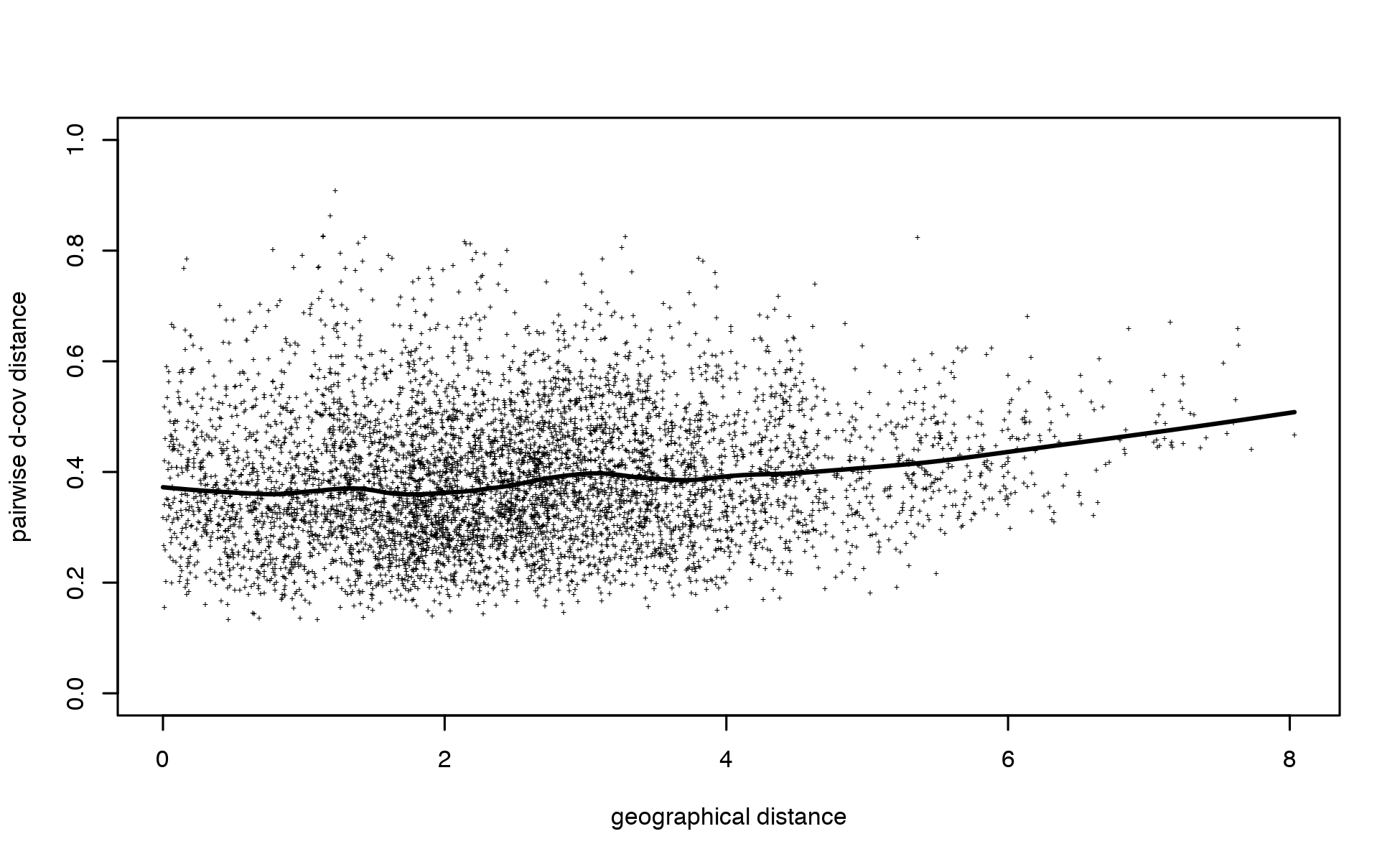}
  \end{center}
  \caption{Scatterplot of the distances between the raw $d_S$-covariances $\breve \Omega_l, \breve \Omega_k$,
    and the corresponding geographical distance between $X_l, X_k$. Notice that the thick line, which
    represents a robust local linear regression obtained via the R function \texttt{lowess}, has a nugget, and
    is slightly increasing with the geographical distance.}
  \label{fig:dcov-dist-vs-geographical-dist}
\end{figure}

\section{Simulation Study}
\label{sec:simulation-study}

In order to quantify whether the spatial mean function and the spatial $d_S$-covariance contain valuable
spatial informations, we compare the results obtained in the paper with a simulation scenario in which
all the spatial locations have the same mean and $d_S$-covariance. 
We simulate observations from a model with constant mean and constant $d_S$-covariance,
\begin{equation}
  \label{eq:global-model}
  Y_{lj}^* = \mu + \vep_{lj}^*, \quad l=1,\ldots, L; j=1,\ldots, n_l
\end{equation}
where $\mu = \left( \sum_{l,j} Y_{lj} \right) / \sum_l n_l$, $\vep_{lj}^*$ were drawn with
replacement from $\left\{ \check \vep_{lj}: l=1,\ldots, L; j=1, \ldots, n_l
\right\}$, $\check \vep_{lj} = \hat \vep_{lj} - \left( \sum_{l,j} \hat \vep_{lj}
\right)/\sum_l n_l$, $\hat \vep_{lj} = Y_{lj} - \hat m(X_l)$, and where $\hat m$ is
the estimated of the mean MFCC field obtained from the data with tuning parameters
$h=0.5, k=14$ nearest locations, and $n_l$ is the number of observations at location
$X_l$.

Notice that although the simulated data is generated under a constant mean model,
their estimated $d_S$-covariance field will be the same as what would be obtained by
a model with varying mean, i.e.\ replacing $\mu$ by $\hat m(X_l)$ in
\eqref{eq:global-model}. Indeed, the $d_S$-covariance field is based on the spatial
smoothing of the sample $d_S$-covariance at each location, defined by 
\begin{align}
  \label{eq:sample-ds-cov}
  \breve \Omega_l^*(t) = \left[ \frac{1}{n_l} \sum_{j=1}^{n_l} \sqrt{ (Y^*_{lj}(t) -
  \overline Y^*_l(t)) (Y^*_{lj}(t) - \overline Y^*_l(t))^\tp } \right]^2,
\end{align}
where $\overline Y^*_l = \sum_j Y^*_{lj}/n_l$. Changing $\mu$ in
\eqref{eq:global-model} to $\hat m(X_l)$ would not change
\eqref{eq:sample-ds-cov}, since
\begin{align*}
  Y_{lj}^* - \overline Y_l^* &= (\mu + \vep^*_{lj}) - \sum_{i} (\mu + \vep^*_{li})/n_l
  \\  &= \vep^*_{lj} - \sum_{i} \vep^*_{li}/n_l
  \\  &= \left( \hat m(X_l) + \vep^*_{lj} \right) - \sum_{i} (\hat m(X_l) + \vep^*_{li})/n_l.
\end{align*}
The $d_S$-covariance field estimated in each simulations run is therefore
the same, regardless of the choice of the mean at each location.

%\begin{equation}
%  \label{eq:global-model}
%  Y^*_{lj} = \mu + \varepsilon_{lj},
%\end{equation}
%  where the distribution of $\varepsilon_{lj}$ does not depend on the spatial location,
%  i.e.  $\varepsilon_{lj} \stackrel{d}{=} \varepsilon = \sum_{i=1}^{200} \sqrt{\lambda_i} \xi_i
%  \varphi_i$, where $\lambda_i$ (respectively $\varphi_i$) is the variance of the $i$ PC scores (respectively
%  the $i$-th PC loading) of the original MFCCs $(Y_{lj}; l=1,\ldots, L; j=1,\ldots, n_j$, and $\xi_1, \ldots, \xi_{200} \stackrel{iid}{\sim} N(0,1)$.
The projections onto PC1-3 are given in Figure~\ref{fig:global-model-mean}. If there was no spatial
information in the mean field of the BNC dataset, the mean field (projected onto PC1) of the simulated data would have the same
range of variation as the mean field of the BNC dataset (projected onto PC1). However, the MFCC field of the
estimated MFCC field of the simulation has consistently a much smaller range than the smooth field obtained
from the BNC dataset over the $100$ simulation replicates (the range for the
projection on PC1 is [9.2, 9.7]
for a realization from \eqref{eq:global-model}, as opposed to [7.9, 10.8] for the real data application). This
provides evidence in support of spatial structure for the mean field.

\begin{figure}[h]
  \begin{center}
    \includegraphics{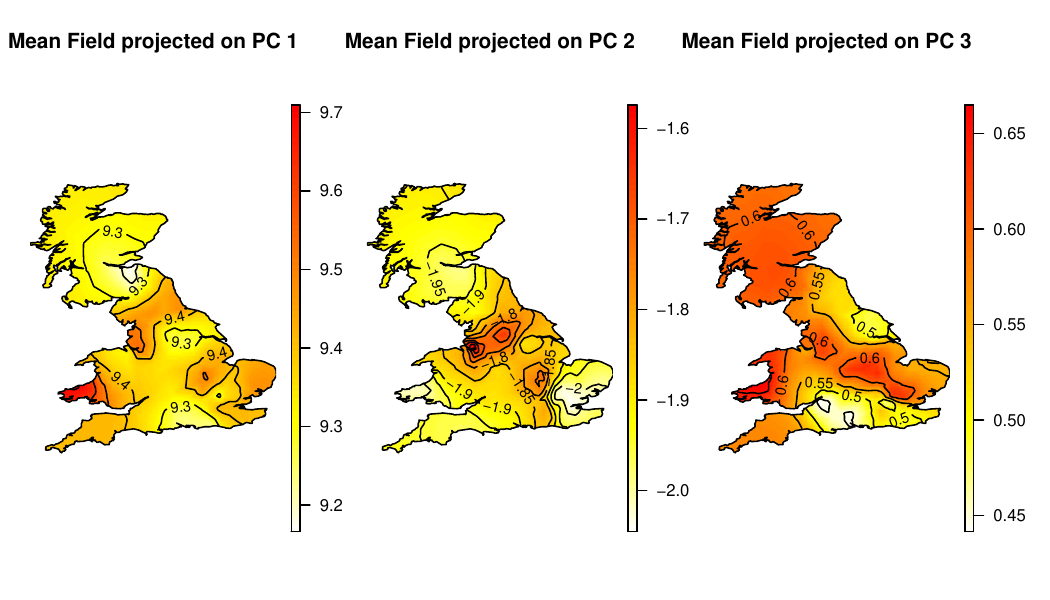}
  \end{center}
  \caption{Projection onto PCs 1,2,3 of the mean obtained from data simulated under the global model \eqref{eq:global-model}}
  \label{fig:global-model-mean}
\end{figure}

\section{An illustration of the advantage of the $d_S$-covariance}
\label{sec:example-d-cov}

As a motivation for the use of $d_S$-covariances, here is a one-dimensional example which illustrates the
advantages of using them when smoothing spatially under the metric $d_S$. Suppose you have data $Y_{11}, \ldots,
Y_{1m} \simiid \vep(x_1)$ and $Y_{21}, \ldots , Y_{2m} \simiid \vep(x_2)$, where $x_1, x_2 \in \bR$ are
two points that are equally close to $x_0 \in \bR$, and we wish to estimate the co-variation of $\vep(x_0)$.
Assume that $\vep(x) \sim N(0, \sigma^2)$ for all $x \in \bR$, and that the mean of $\vep(x)$ is known to be
equal to zero. If we wish to estimate the parameter $\sigma^2 = \var{\vep(x_0)}$, then a natural estimator
is the Fr\'echet mean of $\hat \sigma_i^2 = m^{-1} \sum_{j = 1}^m Y_{ij}^2, i=1,2,$ under $d_S$, i.e.\
\[ \hat \sigma^2_* = \left[ \left( \sqrt{\hat \sigma_1^2} + \sqrt{\hat \sigma_2^2} \right)/2 \right]^2. \]
But 
\begin{align*}
  \ee \hat \sigma^2_* &= \ee(\hat \sigma_1^2 + \hat \sigma_2^2)/4 + \ee\sqrt{\hat \sigma_1^2 \hat \sigma_2^2} / 2
  \\ &< \sigma^2/2 + \sqrt{\ee{\hat \sigma_1^2 \hat \sigma_2^2}}/2
  \\ &= \sigma^2/2 + \sqrt{\ee (\hat \sigma_1^2) \ee (\hat \sigma_2^2)}/2
  \\ &= \sigma^2,
\end{align*}
where we have used Jensen's inequality in the second line (which is in this case a strict inequality, since
$\hat \sigma_1^2 \hat \sigma_2^2$ is not almost surely constant), and the independence of $\hat \sigma_1^2$
and $\hat \sigma_2^2$ in the third line.
In other words, $\hat \sigma^2_*$ is a biased estimator of $\sigma^2$. Furthermore, since $\sqrt{\hat
  \sigma_1^2}$ and $\sqrt{\hat \sigma_2^2}$ are both Chi distributed with $m$ degrees of freedom, $\ee
\sqrt{\hat \sigma^2_*} = \sigma \sqrt{2} \Gamma\left( (m+1)/2 \right)/ \Gamma(m/2)$, where $\Gamma$ is the
Gamma function, $\Gamma(z) = \int _{0}^{\infty }x^{z-1}e^{-x}\,dx$, that is, $\sqrt{ \hat \sigma^2_* }$ is a
biased estimator of $\sqrt{\sigma^2}$. In other words, if one smooths the sample variances using the square-root
Euclidean metric, the resulting estimator is biased, even in the square-root space. However, if one wishes
to estimate the parameter $\tau = \cov_{d_S}(\vep(x_0)) = \left[ \ee{ |\vep(x_0)| } \right]^2$, then the natural estimator is
the Fr\'echet mean of \[ \hat \tau_i = \left( m^{-1} \sum_{j = 1}^m |Y_{ij}| \right)^2, \quad i=1,2\] under $d_S$, that is
\[ \hat \tau_* = \left[ (2m)^{-1} \sum_{j=1}^m \left( |Y_{1j}| + |Y_{2j}| \right)  \right]^2, \] which is
unbiased in the square-root space, i.e.\  $\ee \sqrt{\hat \tau_{*}} = \sqrt{\tau}$. In conclusion, using the
same metric for the spatial smoothing and the definition of the co-variation yields estimators that are less
biased than those obtained by using distinct metrics.

\subsection{Comparison of the $d$-covariance field under the square-root metric and the Euclidean metric}

One might raise the question of whether the $d_S$-covariance field yields results different from the
$d_E$-covariance field ($d_E$ being the Euclidean metric). In order to compare the $d_S$-covariance and
$d_E$-covariance fields visually, one could in principle use dimension reduction methods; however the
interpretation of projections of the $d_S$-covariance may be problematic, as discussed in
Section~\ref{sec:d-cov}.
An alternative way to represent the $d$-covariance variations is to consider a single location of interest
and plot the distances between the $d$-covariance at the
location of interest,  and the $d$-covariances at all other locations of the map. This produces $2D$
surfaces that reflect which parts of the country are more similar or dissimilar to the location of interest
with respect to $d$-covariance. Figure~\ref{fig:sqrt-vs-L2-cov-dist} shows an example of these distance
surface for the square-root and the Euclidean metric, where the distance between $d$-covariances has been
computed using the $d_S$ metric in both cases (averaged over the length of the sound), and the distances have
been renormalized to the interval $[0,1]$ to allow for fair comparison of the plots. The tuning parameters are
$h=1,k=32$ nearest locations.  Notice that the level curves are different. In particular, the
level curve $0.6$ for the square-root map goes down to Bristol, whereas it goes down to Dorset in the
Euclidean metric map. The level curve $0.8$ is also very different between the two maps. These differences can
be attributed to the swelling effect of the Euclidean metric \citep{Arsigny2007}.

\begin{figure}[h]
  \begin{center}
    \includegraphics{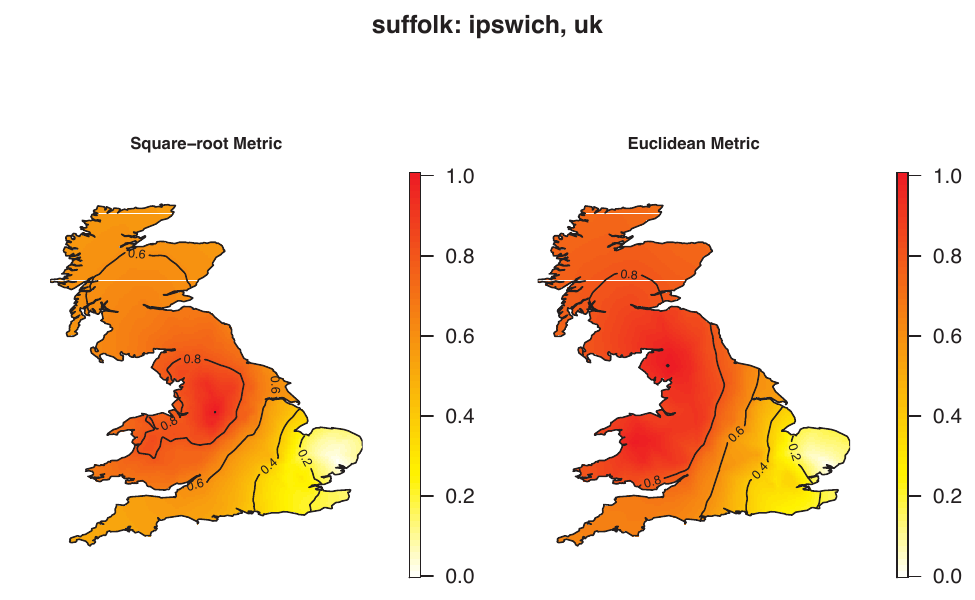}
  \end{center}
  \caption{Left: color map with contours of the pairwise distances between the $d_S$-covariance at Ipswich
    (Suffolk), and the $d_S$-covariance at other locations. Right: same color map, but for the
    $d_E$-covariance. The tuning parameters are $h=1, k=32$ nearest locations. The scale of each map has been
    renormalized so that the value $1$ is the maximal pairwise distance (under the metric $d_S$) in the
    $d$-covariance field, respectively for each metric.}
  \label{fig:sqrt-vs-L2-cov-dist}
\end{figure}

 % \listoftodos

\end{document}